\documentclass[journal]{IEEEtran}
\IEEEoverridecommandlockouts
\usepackage{cite}
\usepackage{amsmath,amssymb,amsfonts}
\usepackage{algorithmic}
\usepackage{graphicx}
\usepackage{textcomp}
\usepackage{subfigure}
\usepackage{float}
\usepackage{xcolor}
\usepackage{enumerate}
\usepackage{pfnote}
\usepackage{makecell}
\usepackage{amsthm}
\usepackage{multirow,tabularx}
\usepackage{caption}
\usepackage{epstopdf}
\usepackage{geometry}
\UseRawInputEncoding
\theoremstyle{plain}
\newtheorem{theorem}{Theorem}
\newtheorem{lemma}{Lemma}

\theoremstyle{definition}

\newtheorem{definition}{Definition}
\newtheorem{remark}{Remark}

\def\BibTeX{{\rm B\kern-.05em{\sc i\kern-.025em b}\kern-.08em
    T\kern-.1667em\lower.7ex\hbox{E}\kern-.125emX}}

\usepackage{amsmath}
\allowdisplaybreaks[3]
\usepackage[colorlinks, linkcolor=blue, anchorcolor=blue, citecolor=blue]{hyperref}

\makeatletter
\renewcommand{\maketag@@@}[1]{\hbox{\m@th\normalsize\normalfont#1}}%
\newcommand{\indicator}[1]{\mathbf{1}\left\{ #1 \right\}}
\makeatother

\begin{document}
\onecolumn

\title{
   Secure Joint Source-Channel Coding for the AWGN Channel with Feedback: A Finite Blocklength Analysis
	}

\author{
\IEEEauthorblockN{
Sheng~Su\IEEEauthorrefmark{1},
Yuhan~Yang\IEEEauthorrefmark{1},
Chao~Qi\IEEEauthorrefmark{2},
Xuan~He\IEEEauthorrefmark{1},
Bin~Dai\IEEEauthorrefmark{1},
Xiaohu~Tang\IEEEauthorrefmark{1}}

\IEEEauthorblockA{
\IEEEauthorrefmark{1}
School of Information Science and Technology, Southwest Jiaotong University, Chengdu, 610031, China.}

\IEEEauthorblockA{
\IEEEauthorrefmark{2}
School of Computer and Software Engineering, Xihua University, Chengdu, 610039, China.}

\{ache, yangyvhan\}@my.swjtu.edu.cn, chaoqi@mail.xhu.edu.cn, \{xhe, daibin, xhutang\}@swjtu.edu.cn.
}

\maketitle

\begin{abstract}
In the literature, it has been shown that the secrecy capacity of the additive white Gaussian noise (AWGN) wiretap channel with noise-free feedback equals the capacity of the same model without secrecy constraint, and the classical Schalkwijk-Kailath (SK) scheme achieves the secrecy capacity.
In this paper, we show that in finite blocklength regime, the SK scheme is not optimal, and propose a modified SK scheme which may perform better than the classical one. Besides this,
this paper establishes a finite blocklength converse for the AWGN wiretap channel with feedback, which can also be viewed as a converse for the same model without secrecy constraint. To the best of the authors' knowledge, this is the first paper to address such a problem, and the results of this paper are further explained via numerical examples.

\end{abstract}

\section{Introduction}\label{sec1}
\subsection{Problem formulation}\label{sec1-1}
\begin{figure}[htb]
	\centering
	\includegraphics[scale=0.7]{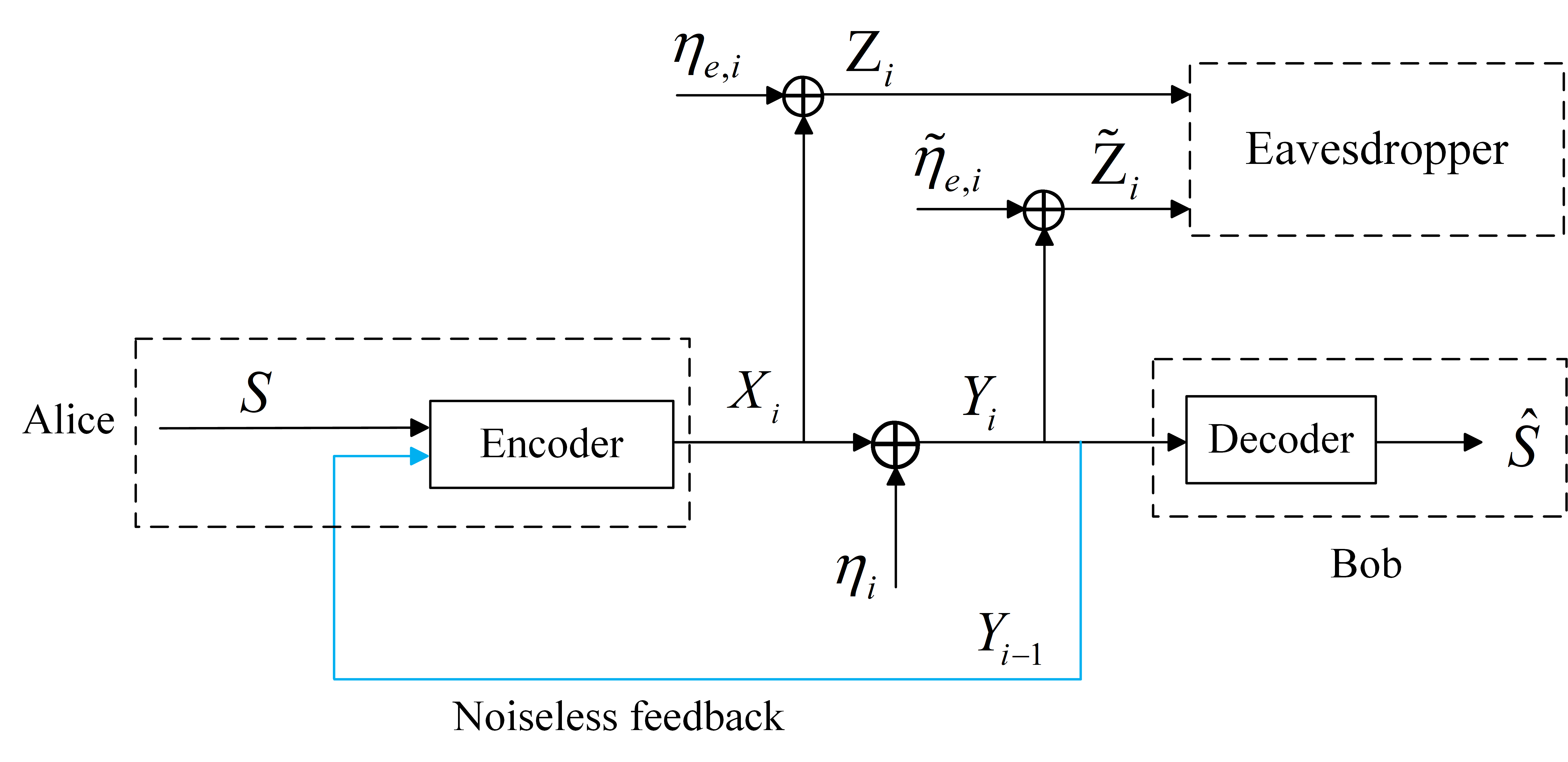}
	\caption{Joint source-channel coding (JSCC) for the Gaussian wiretap channel with noise-free feedback}
	\label{fg01}
\end{figure}
In Figure \ref{fg01}, Alice wishes to transmit a source $S\sim \mathcal{N}(0, \sigma_{s}^2)$ \footnote{In the literature, it often considers the transmission of an independent identical distributed (i.i.d.) source sequence $S^k$ with length $k$, and the rate of the JSCC encoder is defined as $\frac{kH_{s}}{N}$, where $H_{s}$ is the entropy of the source component. To maximize the JSCC rate equals to find a minimum coding blocklength $N$ when $k$ and $H_{s}$ are fixed.
While this paper focuses on a simplified setup, namely, for one symbol source $k=1$, what is the minimum coding blocklength $N$ achieving the largest JSCC rate $\frac{1}{N}$.}
over the additive white Gaussian noise (AWGN) channel to Bob in the leverage of noise-free channel feedback, and Bob
wishes to reconstruct the source $S$ with a tolerable distortion. In the meanwhile, an eavesdropper attempts to eavesdrop $S$ via eavesdropping channels.

At time instant $i\in \{1,2,\ldots,N\}$, the channel inputs-outputs relationships are given by
{\small \begin{equation}\label{1to1-model1}
	\begin{split}
	&Y_{i}=X_{i}+\eta_{i},\\
	&Z_{i}=X_{i}+\eta_{e,i},\,\,\,
	\widetilde{Z}_{i}=X_{i}+\eta_{i}+\widetilde{\eta}_{e,i},
	\end{split}
\end{equation}}%
where
$X_{i}$ is the output of the JSCC encoder,
$Y_{i}$ is Bob's received signal,
$\eta_{i}\sim \mathcal{N}(0, \sigma_{\eta}^2)$ is Bob's channel noise,
$Z_{i}$, $\widetilde{Z}_{i}$ are the outputs of the eavesdropping channels,
$\eta_{e,i}\sim \mathcal{N}(0, \sigma_{e}^2)$, $\widetilde{\eta}_{e,i}\sim \mathcal{N}(0, \widetilde{\sigma}_{e}^2)$
are the corresponding eavesdropping channel noises.
Here $\eta_{i}$, $\eta_{e,i}$ and $\widetilde{\eta}_{e,i}$ are independent identical distributed (i.i.d.) across the time index $i$,
and $S$, $\eta_{i}$, $\eta_{e,i}$, $\widetilde{\eta}_{e,i}$ are independent of one another.

\begin{definition}\label{def01}
	An $(N,P)$-feedback code with average power constraint consists of:
	\begin{itemize}
		\item At time $i \in \{1,2,\ldots,N\}$,
			the JSCC encoder encodes the source $S$, past feedback signals $Y^{i-1}=(Y_{1}, \ldots, Y_{i-1})$, and as the codeword $X_{i}= \varphi_{i} (S, Y^{i-1})$, where $\varphi_{i}$ is the encoding function.
			Define the average power constraint of the codeword $X^{N}=(X_{1}, \ldots, X_{N})$ by
{\small			\begin{equation}\label{1to1-P-a}
				\frac{1}{N}\sum_{i=1}^{N}\mathbb{E} [X_{i}^{2}] \leq P.
			\end{equation}}%
		\item At time $i \in \{1,2,\ldots,N\}$, a decoder outputs $\widehat{S}_{i}=\psi(Y^{i})$,
			where $\widehat{S}_{i}$ is the decoder's estimation of $S$ at time $i$, and $\psi$ is the decoding function.
	\end{itemize}
\end{definition}

To evaluate the decoding performance,
the excess-distortion probability \cite{FBL-JSCC} at time $i \in \{1,2,\ldots,N\}$ is defined as
{\small \begin{equation}\label{1to1-Pdi}
    \mathcal{P}_{d,i}=\mathbb{P} \left[ (S - \widehat{S}_{i})^{2}\geq d\right],
\end{equation}}%
where $0<d<\sigma_{s}^{2}$ is a distortion threshold.
Furthermore, $P_{d,N}$ is Bob's excess-distortion probability at the last index $N$.

For the eavesdropper, the information leakage rate \cite{secure2}-\cite{secure1} is defined as
{\small \begin{equation}\label{1to1-Ln-s}
L_{N} = \frac{1}{N} I(S; Z^{N}, \widetilde{Z}^{N}),
\end{equation}}%
and if the leakage rate satisfies $L_{N} \leq \delta$ $(\delta>0)$, we say
$\delta$-secrecy is achieved.
Here note that in (\ref{1to1-Ln-s}),
$\delta=0$ corresponds to the weak secrecy extensively studied in the literature.

Given distortion threshold $0<d<\sigma_{s}^{2}$, the targeted excess-distortion probability $\varepsilon>0$ and secrecy threshold $\delta>0$,
there exists a minimum coding blocklength $N$ and a pair of JSCC encoder-decoder such that
{\small \begin{equation}\label{1to1-R-d}
\begin{split}
	\mathcal{P}_{d, N} \leq \varepsilon,\,\,\,\, L_{N}\leq \delta,
\end{split}
\end{equation}}%
the corresponding $(d, \varepsilon)$-JSCC rate $\frac{1}{N}$ is called the $\delta$-secrecy capacity $\mathcal{C}^{\rm f}_{\rm s}(d,\varepsilon,\delta)$. Bounds on $\mathcal{C}^{\rm f}_{\rm s}(d,\varepsilon,\delta)$ will be given in the next section.

\subsection{Literature review}\label{sec1-2}

Wyner in his pioneering work \cite{secure} first established the fundamental limit of secure transmission over a physically degraded wiretap channel, and showed that the secrecy capacity suffers a rate loss due to weak secrecy constraint. Though a pair of JSCC encoder-decoder is considered in the model of the wiretap channel
\cite{secure}, the secrecy capacity is characterized by using a source-channel separation scheme,
and the JSCC schemes were considered in \cite{secure-jscc-1}-\cite{secure-jscc-2}.

From channel coding aspect, though channel feedback does not increase the capacity of a memoryless channel, \cite{secure2} showed that the remarkable Schalkwijk-Kailath (SK) scheme \cite{SK}, which achieves the capacity of the AWGN channel with noise-free feedback, satisfies weak secrecy by itself. This also indicates that the secrecy capacity of the AWGN wiretap channel with noise-free feedback equals that of the same model without secrecy constraint, and it is achieved by the SK scheme. However, \cite{FBL-feedback} showed that in the finite blocklength (FBL) regime, the SK scheme
is not optimal for the AWGN channel, which also indicates that it may not be optimal for the
AWGN wiretap channel in FBL regime.

In recent years, \cite{SK-v} showed that the SK scheme can be directly extended to a JSCC setup, and it is not difficult to check that the self-secure property \cite{secure2} of the SK scheme also holds for this JSCC setup. However, in the JSCC setup of the AWGN channel with noise-free feedback,
whether the extended SK scheme of \cite{SK-v} is optimal remains unknown since
the FBL converse for such a model is missing \footnote{To the best of the authors' knowledge, the FBL converse for the JSCC setup was only considered in non-feedback case \cite{FBL-JSCC}.}.

\subsection{Contribution of this paper}\label{sec1-3}

In this paper, we establish a finite blocklength converse for the AWGN wiretap channel with feedback, which can also be viewed as a converse for the same model without secrecy constraint.
Besides this,
in finite blocklength regime, we show that the JSCC based SK scheme \cite{SK-v} is not optimal, and propose a modified SK scheme which may perform better than the existing one \cite{SK-v}.

\subsection{Notation}\label{sec1-4}

The ceiling function is denoted by $\lceil \cdot \rceil$,
and define $\exp(x) \triangleq e^{x}$.
The complementary Gaussian cumulative distribution function ($Q$-function) is denoted by {\small $Q(x) = \frac{1}{\sqrt{2\pi}} \int_{x}^{\infty} \exp\left(-\frac{t^2}{2}\right) dt$}, with $Q^{-1}(\cdot)$ being its inverse.
The Gaussian capacity function is {\small $C(x) \triangleq \frac{1}{2} \ln\left(1 + \frac{x}{\sigma_{\eta}^2}\right)$}, while the rate-distortion function for source $S$ under squared-error distortion is {\small $R(x) \triangleq  \frac{1}{2} \ln\left( \frac{\sigma_s^2}{x} \right) $}.
The Gaussian dispersion function is given by {\small $V(x) \triangleq \frac{x(x + 2\sigma_{\eta}^2)}{2(x + \sigma_{\eta}^2)^2}$}, and $V_d=\frac{1}{2}$ denotes the dispersion for the Gaussian source $S$.

\section{Preliminaries}\label{sec2}
In this section, we introduce the JSCC based SK scheme \cite{SK-v} for the AWGN channel with noise-free feedback. The problem setup is the same as the model of Figure \ref{fg01} without the consideration of the eavesdropper, and this scheme
is described below.

At time instant $1$, the transmitter sends
{\small \begin{equation}\label{pr-sk-xh1}
		X_{1}=\sqrt{\frac{P}{\sigma_{s}^{2}}}S,
\end{equation}}%
where the power of $X_1$ is $P$. Meanwhile, the receiver obtains
$Y_1=X_1+\eta_{1}$, and computes his first estimation of $S$ by
{\small \begin{equation}\label{pr-sk-es1}
	\widehat{S}_{1}=Y_{1}\cdot\left(\sqrt{\frac{P}{\sigma_{s}^{2}}}\right)^{-1}= S +\eta_{1}\cdot\left(\sqrt{\frac{P}{\sigma_{s}^{2}}}\right)^{-1}\triangleq S+\epsilon_{1},
\end{equation}}%
where {\small $\epsilon_{1}  \triangleq \widehat{S}_{1}-S=\eta_{1}\cdot \left(\sqrt{\frac{P}{\sigma_{s}^{2}}}\right)^{-1}$} is  the first estimation error, and its variance is $\alpha_{1}\triangleq {\rm Var}(\epsilon_{1})=\sigma_{\eta}^{2}\sigma_{s}^{2}\cdot P^{-1}$.
At time instant $i$ ($i\in\{2,3,\dots,N\}$), the transmitter perfectly observes
$Y^{i-1}=(Y_{1},\ldots,Y_{i-1})$, calculates the previous time instant's estimation $\widehat{S}_{i-1}$ of the receiver, and then sends the estimation error $\epsilon_{i-1}\triangleq\widehat{S}_{i-1}-S$ by
{\small \begin{equation}\label{pr-sk-xi}
	X_{i}=\sqrt{\frac{P}{\alpha_{i-1}}}\epsilon_{i-1},
\end{equation}}%
where $\alpha_{i-1} \triangleq {\rm Var}(\epsilon_{i-1})$ and the power of $X_{i}$ is ${\rm E}[X^{2}_{i}]=P$.
Meanwhile, the receiver obtains  $Y_{i}=X_{i}+\eta_{i}$, and updates the $i$-th estimation of $S$ by
{\small \begin{equation}\label{pr-sk-Si}
	\widehat{S}_{i}=\widehat{S}_{i-1}-\beta_{i}Y_{i} \triangleq  \widehat{S}_{i-1}-\widehat{\epsilon}_{i-1},
\end{equation}}%
where $\widehat{\epsilon}_{i-1}\triangleq \beta_{i}Y_{i}$ is the Minimum Mean Square Estimation (MMSE) estimation of $\epsilon_{i-1}$,
and thus the MMSE coefficient is given by
{\small \begin{equation}\label{pr-sk-betai}
	\beta_{i}=\frac{{\rm E}[Y_{i}\epsilon_{i-1}]}{{\rm E}[Y^{2}_{i}]}\overset{(a)}=\frac{\sqrt{P\alpha_{i-1}}}{P+\sigma_{\eta}^2},
\end{equation}}%
where $(a)$ follows from $Y_{i}=X_{i}+\eta_{i}$ and (\ref{pr-sk-xi}).
The updated estimation error is
{\small \begin{equation}\label{pr-sk-epi}
	\epsilon_{i}  \triangleq \widehat{S}_{i}-S=\epsilon_{i-1}-\beta_{i}Y_{i},
\end{equation}}%
and its variance is calculated as
{\small \begin{equation}\label{pr-sk-ai}
	\begin{aligned}
	&\alpha_{i}\triangleq {\rm Var}(\epsilon_{i})=\frac{\sigma_{\eta}^2\sigma_{s}^2}{P}
(\frac{\sigma_{\eta}^2}{P+\sigma_{\eta}^2})^{i-1}.
	\end{aligned}
\end{equation}}%
At time instant $N$, from (\ref{pr-sk-Si}), the receiver computes the final estimation  of $S$ as
{\small \begin{equation}\label{ctn}
		\widehat{S}_{N}=S+\epsilon_{N},
\end{equation}}%
where $\epsilon_{N}$ is the final estimation error.

\begin{lemma}{(Lower bound on $\mathcal{C}^{\rm f}_{\rm s}(d,\varepsilon,\delta)$): }\label{l1}
	Following the JSCC based SK scheme \cite{SK-v} described above,
	for given distortion threshold $d$, targeted excess-distortion probability $\varepsilon>0$ and secrecy threshold $\delta>0$,
	the $\delta$-secrecy capacity $\mathcal{C}^{\rm f}_{\rm s}(d,\varepsilon,\delta)$ is lower bound by
{\small \begin{equation}\label{rn-sk-c}
	\begin{aligned}
			 \mathcal{C}^{\rm f}_{\rm s}(d,\varepsilon,\delta) \geq \min\{\frac{1}{\widetilde{N_{1}}},\,\frac{1}{\widetilde{N_{2}}}\},
	\end{aligned}
\end{equation}}%
where
{\small \begin{equation}\label{n1-skc}
	\begin{aligned}
	&\widetilde{N_{1}}=\left\lceil \frac{1}{C(P)}\left(R(d)+\ln\left( Q^{-1}\left(\frac{\epsilon}{2}\right) \right)+\frac{1}{2}\ln\left(1+\frac{\sigma_{\eta}^{2}}{P}\right)\right)\right\rceil,\\
	&\widetilde{N_{2}}=\left\lceil \frac{1}{2\delta}\ln\left(\left(1+\frac{P}{\sigma_{e}^{2}}\right)\cdot \left(1+\frac{P}{\widetilde{\sigma}_{e}^{2}}\right)\right)\right\rceil.
\end{aligned}
\end{equation}}%
\end{lemma}

\begin{IEEEproof}
Here note that $\widetilde{N_{1}}$ is determined by the legal receiver Bob's decoding error probability, and $\widetilde{N_{2}}$ is by the analysis of the eavesdropper's information leakage rate.
The proof is similar to the proof of \cite{secure2}, and we omit it here.
\end{IEEEproof}

In the next section, we propose an upper bound on $\mathcal{C}^{\rm f}_{\rm s}(d,\varepsilon,\delta)$ and give a new lower bound by modifying the above mentioned JSCC based SK scheme \cite{SK-v}.

\section{New Bounds on $\mathcal{C}^{\rm f}_{\rm s}(d,\varepsilon,\delta)$}\label{sec3}

Before formal results are presented, define
{\small \begin{equation}\label{F1x}
	\begin{split}
		 F_{1}(x) &\triangleq x\cdot C\left(\frac{P}{1 - \varepsilon }\right) +\sqrt{V_{d}+ x\cdot V\left(\frac{P}{1-\varepsilon}\right)}\cdot \sqrt{\ln(x)} +O\left(\sqrt{x}\right) ,\\
		 F_{2}(x)  &\triangleq\frac{1}{2x}\left(\frac{(P+\sigma_{\eta}^{2})(\sigma_{e}^{2}+\widetilde{\sigma}_{e}^{2})}{\sigma_{e}^{2}\widetilde{\sigma}_{e}^{2}}\right)\left(1-\left( \frac{\sigma_{\eta}^{2}}{P+\sigma_{\eta}^{2}} \right)^{x}\right),\\
	\end{split}
\end{equation}}%
which will be used in the remainder of this section.

\subsection{Main Results}\label{sec3-1}
\begin{theorem}{(Upper bound on $\mathcal{C}^{\rm f}_{\rm s}(d,\varepsilon,\delta)$): }\label{th01}
	For given distortion threshold $d$, targeted excess-distortion probability $\varepsilon>0$ and secrecy threshold $\delta>0$,
	an upper bound on $\mathcal{C}^{\rm f}_{\rm s}(d,\varepsilon,\delta)$ is given by
{\small \begin{equation}\label{upperbound}
	\begin{aligned}
			\mathcal{C}^{\rm f}_{\rm s}(d,\varepsilon,\delta) \leq \frac{1}{N_{1}-1},
	\end{aligned}
\end{equation}}%
where $N_{1}\geq 2$ satisfies $F_{1}(N_{1}-1)<R(d)$ and $F_{1}(N_{1})\geq R(d)$.
\end{theorem}

\begin{IEEEproof}
Since the secrecy capacity of the AWGN wiretap channel with noise-free feedback equals that of the same model without secrecy constraint \cite{secure2}, the FBL secrecy capacity is at least upper bounded by the FBL capacity of the same model without secrecy constraint.
Observing that \cite{FBL-JSCC} investigated the FBL capacity of the AWGN channel in a JSCC setup, and \cite{FBL-feedback} studied the FBL capacity of the AWGN channel with noise-free feedback in channel coding setup. Combining the converse tools used in \cite{FBL-JSCC} and \cite{FBL-feedback}, we derive
a converse bound on the FBL capacity of the AWGN channel with noise-free feedback in a JSCC setup.
The detailed proof is given in Section \ref{sec4-1}.
\end{IEEEproof}

\begin{theorem}{(New lower bound on $\mathcal{C}^{\rm f}_{\rm s}(d,\varepsilon,\delta)$): }\label{th02}
	For given distortion threshold $d$, targeted excess-distortion probability $\varepsilon>0$ and secrecy threshold $\delta>0$,
	$\mathcal{C}^{\rm f}_{\rm s}(d,\varepsilon,\delta)$ is lower bounded by
{\small \begin{equation}\label{rn1}
	\begin{aligned}
			&\mathcal{C}^{\rm f}_{\rm s}(d,\varepsilon,\delta) \geq \min\{\frac{1}{N_{2}},\,\frac{1}{N_{3}}\},
	\end{aligned}
\end{equation}}%
where
{\small \begin{equation}\label{n2}
	\begin{aligned}
	N_{2}=\left\lceil \frac{1}{C(P)}\left(R(d)+\ln\left( Q^{-1}\left(\frac{\epsilon}{2}\right) \right)\right)\right\rceil
	\end{aligned}
\end{equation}}%
and
$N_{3}\geq 1$ satisfies $F_{2}(N_{3}-1)>\delta$ and $F_{2}(N_{3})\leq\delta$.

\end{theorem}

\begin{IEEEproof}
This new lower bound is based on a slight modification of the JSCC based SK scheme \cite{SK-v}. In particular, at time instant $1$, the MMSE strategy is used to replace the zero-forcing method of \cite{SK-v}, and the remainder of the encoding-decoding procedure keeps the same. Note that $N_{2}$ is determined by Bob's decoding error probability, and $N_{3}$ is by the analysis of the eavesdropper's information leakage rate, which is significantly different from that of \cite{secure2} since the information leakage occurs at every time instant, while that only occurs at the first time instant in \cite{secure2}.
The detailed proof is given in Section \ref{sec4-2}.
\end{IEEEproof}

\begin{remark}
It is not difficult to check that our modified SK scheme in Theorem \ref{th02} also satisfies weak secrecy, namely, $\lim\limits_{N\to +\infty}L_{N}=0$.
\end{remark}

\subsection{Numerical Example}\label{sec3-2}
\begin{figure}[htb]
	\centering
	\subfigure{\includegraphics[scale=0.61]{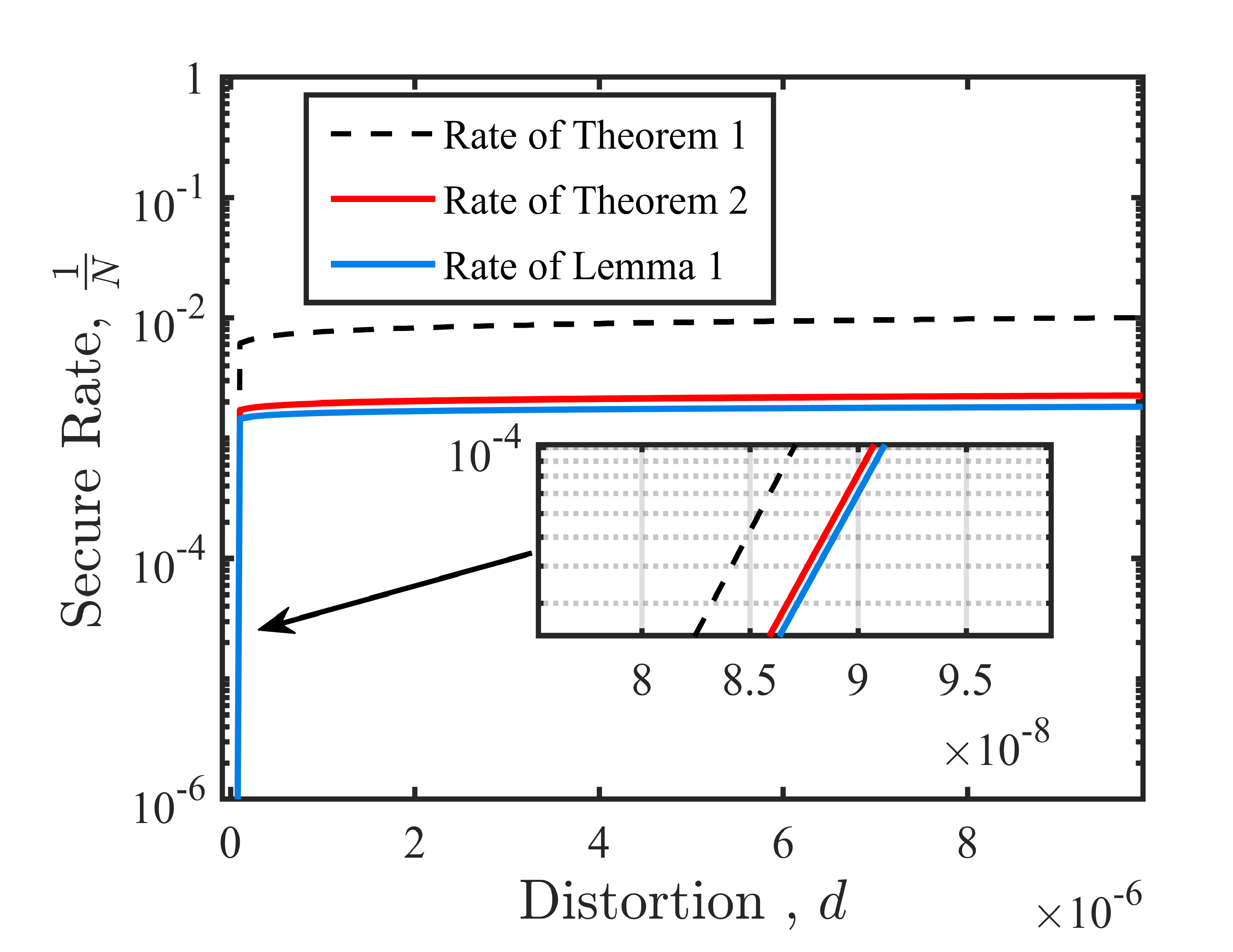}}
	\subfigure{\includegraphics[scale=0.52]{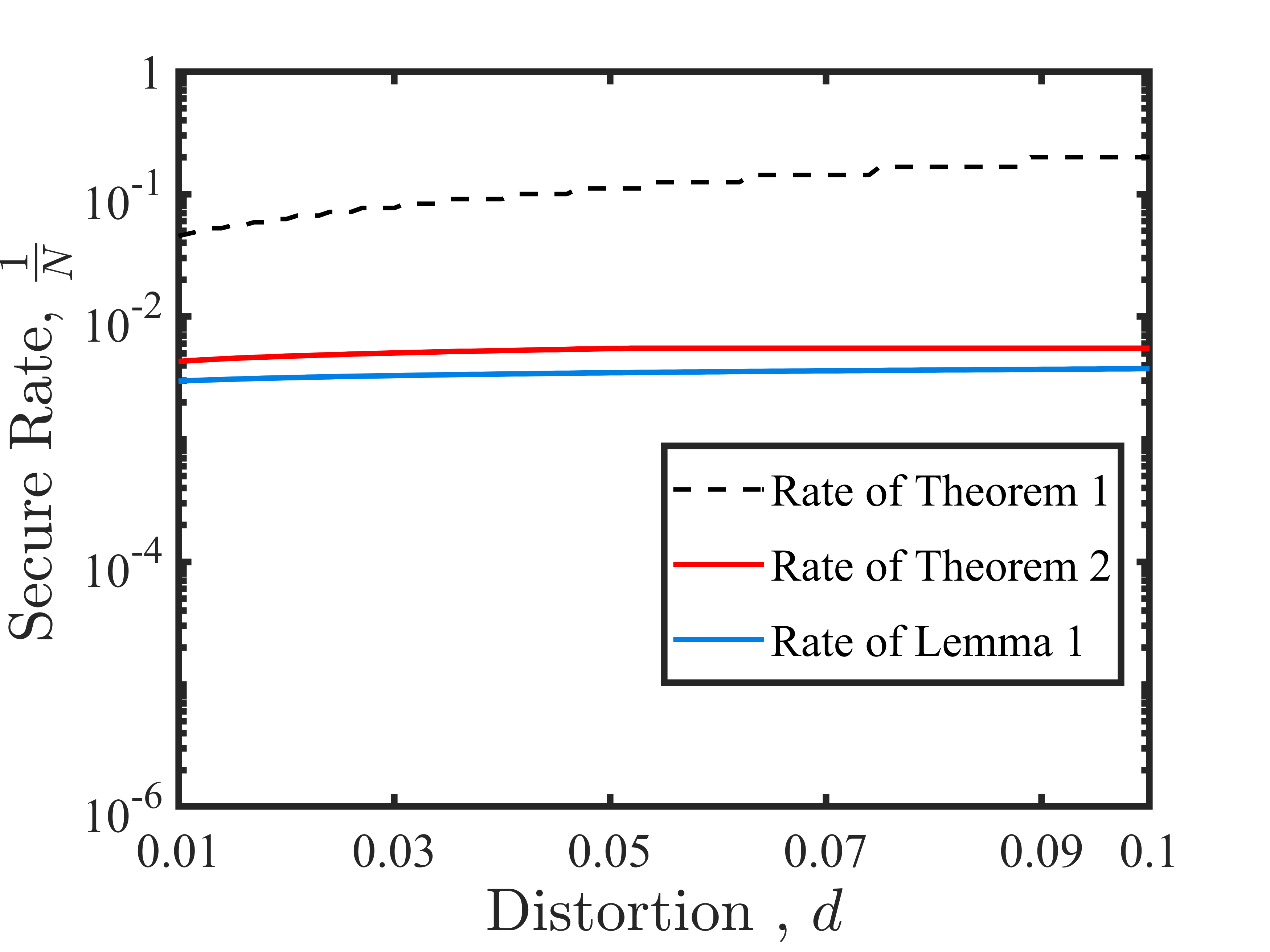}}
	\caption{Rate versus distortion for $\sigma_{s}^{2}=1$, $\sigma_{\eta}^{2}=30$, $\sigma_{e}^{2}=30$, $\widetilde{\sigma}_{e}^{2}=40$, $\varepsilon=10^{-5}$, $\delta=0.01$ and $P=1$}
	\label{fg02}
\end{figure}
Figure \ref{fg02} shows that our modified SK scheme performs better than the classic one \cite{SK-v}
when the signal-to-noise ratio is low, and this is due to fact that the MMSE strategy adopted by the receiver at initial time enhances the receiver's decoding performance, and though this modification causes additional information leakage to the eavesdropper, the average leakage rate still vanishes as coding blocklength tends to infinity. Besides this, it is shown that
the gap between upper and lower bounds is increasing while the distortion is increasing, and how to eliminate this gap will be our future work.

\section{Proofs}\label{sec4}

\subsection{Proof of Theorem \ref{th01}}\label{sec4-1}
The proof is given by the following steps,
{\small \begin{equation}\label{NKgoal}
	\begin{split}
		\mathcal{C}^{\rm f}_{\rm s}(d,\varepsilon,\delta)
		&\overset{(a)}{\leq} \mathcal{C}_{\text{fb, ap}}(P, d, \varepsilon)= \left(N_{\text{fb, ap}}^{*}(P, d, \varepsilon)\right)^{-1}\\
		&\overset{(b)}{\leq}\left( N_{\text{fb, pp}}^{*}\left(P^{\prime}, d, \varepsilon^{\prime} \right) \right)^{-1}\\
		&\overset{(c)}{\leq}\left( N_{\text{fb, ep}}^{*}\left(P^{\prime}, d, \varepsilon^{\prime} \right) -1\right)^{-1}\\
		&\overset{(d)}{=} \frac{1}{N_{1}-1}.
	\end{split}
\end{equation}}%
Here note that
\begin{itemize}
	\item $P^{\prime} \triangleq \frac{P}{1 - \varepsilon -\frac{\zeta}{\sqrt{N}} }$ and $\varepsilon^{\prime} \triangleq 1 - \frac{\zeta}{\sqrt{N}} $, where $\zeta > 0$ is an auxiliary parameter.
    \item $\mathcal{C}_{\text{fb, ap}}(P, d, \varepsilon)$ is the capacity of an AWGN channel with noiseless feedback, under \emph{average} power constraint $P$, distortion threshold $d$, and excess-distortion probability $\varepsilon$.
	$N_{\text{fb, ap}}^{*}(P, d, \varepsilon)$ is the minimum coding blocklength that achieves $\mathcal{C}_{\text{fb, ap}}(P, d, \varepsilon)$.
    \item $N_{\text{fb, pp}}^{*}(P', d, \varepsilon')$ is the minimum coding blocklength under \emph{peak}  power constraint $P^{\prime}$.
    \item $N_{\text{fb, ep}}^{*}(P', d, \varepsilon')$ is the minimum coding blocklength under \emph{equal}  power constraint $P^{\prime}$.
    \item $N_{1}$ is the minimum blocklength derived from (\ref{c-goal}).
\end{itemize}
Steps $(a)-(d)$ are proved as follows
\subsubsection{\textbf{Proof of step (a)}}\label{sec4-1-1}
In the asymptotic regime, \cite{secure2} showed that the secrecy capacity $\mathcal{C}_{s}^{f}$ equals $\mathcal{C}^{f}$, namely,
{\small \begin{equation}\label{sb-1}
	\mathcal{C}_{s}^{f}=\mathcal{C}^{f}=\frac{1}{2}\ln\left(1+\frac{P}{\sigma_{\eta}^{2}}\right).
\end{equation}}%
In finite blocklength regime, whether (\ref{sb-1}) still holds or not remains unknown, however, the FBL form of $\mathcal{C}^{f}$ is obvious to serve as an upper bound on the FBL form of $\mathcal{C}_{s}^{f}$,
namely, $\mathcal{C}^{\rm f}_{\rm s}(d,\varepsilon,\delta) \leq \mathcal{C}_{\text{fb,ap}}(P, d, \varepsilon)$, where $\mathcal{C}_{\text{fb,ap}}(P, d, \varepsilon)$ is defined as follows.
\begin{definition}\label{def03-mod}
	Consider an $(N, P)$-feedback code in Definition \ref{def01}. For given source {\small $S\sim \mathcal{N}(0, \sigma_s^2)$} and distortion threshold $d$, if
	the excess-distortion probability {\small $\mathbb{P} \left[ (S-\widehat{S}_{N})^{2} > d\right]\leq \varepsilon$},
	an $(N, P)$-feedback code is called an {\small $(N, P, d, \varepsilon)$}-feedback code.

	For fixed $\varepsilon$, $d$ and average power constraint $P$, define the minimum coding blocklength as
{\small 	\begin{equation}
		N_{\text{fb, ap}}^{*}(P, d, \varepsilon) = \inf \left\{ N \in \mathbb{Z}^+ : \exists \text{ an } (N, P, d, \varepsilon)\text{-feedback code} \right\}.
	\end{equation}}%
	The corresponding capacity is given by
{\small	\begin{equation}
		C_{\text{fb, ap}}(P, d, \varepsilon) = (N_{\text{fb,ap}}^{*}(P, d, \varepsilon))^{-1}.
	\end{equation}}%
\end{definition}
Since any code achieving $\mathcal{C}^{\rm f}_{\rm s}(d,\varepsilon,\delta)$ must be an $(N, P, d, \varepsilon)$-feedback code, we have
{\small \begin{equation}\label{eq-secrecy-upper-bound}
	\begin{split}
		\mathcal{C}^{\rm f}_{\rm s}(d,\varepsilon,\delta) &\leq C_{\text{fb,ap}}(P, d, \varepsilon) = \left( N_{\text{fb,ap}}^{*}(P, d, \varepsilon) \right)^{-1},
	\end{split}
\end{equation}}%
which completes the proof of step $(a)$.

\subsubsection{\textbf{Proof of step (b)}}\label{sec4-1-2}
Given an $(N, P, d, \varepsilon)$-feedback code with encoders $X_{i} = \varphi_{i}(S, Y^{i-1})$ ($i = 1, \ldots, N$) and decoder $\widehat{S} = \psi(Y^N)$ satisfying the average power constraint,
similar to \cite{FBL-feedback}, we can construct a new code with the same coding blocklength $N$ satisfying a peak power constraint.
Specifically,
for $\zeta > 0$, define a new encoding function as follows.
\begin{itemize}
	\item For $i = 1, \ldots, N$, define the new encoding functions,
{\small \begin{equation}\label{Xprime}
		X_{i}^{\prime} = X_{i} \cdot \indicator{ \sum_{j=1}^{i} X_{j}^2 \leq NP^{\prime}},
	\end{equation}}%
	where $\indicator{\cdot}$ is the indicator function and $P^{\prime} \triangleq \frac{P}{1 - \varepsilon - \frac{\zeta}{\sqrt{N}}}$.
	\item The decoding function remains $\widehat{S} = \psi(Y^N)$.
\end{itemize}
We can check that
{\small \begin{equation}\label{peak-power}
	\begin{split}
	&\mathbb{P}\left[\sum_{i=1}^{N} (X_{i}^{\prime})^2 \leq N P^{\prime}\right]
	=\mathbb{E}\left[ \indicator{\sum_{i=1}^{N} (X_{i}^{\prime})^2 \leq N P^{\prime}}\right]
	=\int_{{x^{\prime}}^{N}} f({x^{\prime}}^{N})\cdot \indicator{\sum_{i=1}^{N} (x_{i}^{\prime})^2 \leq N P^{\prime}} d{x^{\prime}}^{N}
	\overset{(a^{\prime})}{=}\int_{{x^{\prime}}^{N}} f({x^{\prime}}^{N})\cdot 1 \cdot d{x^{\prime}}^{N}=1,
	\end{split}
\end{equation}}%
where 
$f(\cdot)$ denotes probability density function,
and $(a^{\prime})$ follows from 
{\small \begin{equation}\label{xsum}
	\begin{split}
	\sum_{i=1}^{N} (x_{i}^{\prime})^2= \sum_{i=1}^{N}\left(x_{i} \cdot \indicator{ \sum_{j=1}^{i} x_{j}^2 \leq NP^{\prime}}\right)^{2} \leq N P^{\prime}.
\end{split}
\end{equation}}%
Therefore, the new code is under peak power constraint $P^{\prime}$.
Its excess-distortion probability $\mathcal{P}_{d, N}$ is upper bounded by
{\small \begin{align}\label{edp}
		&\mathcal{P}_{d, N}\overset{(b^{\prime})}{\leq} 1 - \frac{\zeta}{\sqrt{N}} \,\,\triangleq\,\, \varepsilon^{\prime},
\end{align}}%
where the proof of $(b^{\prime})$ is in Appendix \ref{App-28-edpu}.

Hence, for any $\zeta > 0$, any $(N, P, d, \varepsilon)$-code implies the existence of a code under peak power $P^{\prime}$ and excess-distortion probability at most $\varepsilon^{\prime} $.
Consequently,
{\small \begin{equation}
	N_{\text{fb, ap}}^{*}(P, d, \varepsilon) \geq N_{\text{fb, pp}}^{*}\left( P^{\prime}, d, \varepsilon^{\prime} \right),
\end{equation}}%
where $N_{\text{fb, pp}}^{*}(P^{\prime}, d, \varepsilon^{\prime})$ is the minimum coding blocklength for an $(N, P^{\prime}, d, \varepsilon^{\prime})$-feedback code under peak power constraint $P^{\prime}$, which completes the proof of step $(b)$.

\subsubsection{\textbf{Proof of step (c)}}\label{sec4-1-3}
Consider an $(N, P^{\prime}, d, \varepsilon^{\prime})$-feedback code with peak power constraint
{\small \begin{equation}\label{ppic}
	\mathbb{P}\left[ \sum_{i=1}^{N} X_{i}^2 \leq NP^{\prime} \right] = 1.
\end{equation}}%
Let $N^{\prime}=N+1$, which implies that $N^{\prime}\geq 2$, by appending an extra channel use with a carefully chosen input $X_{N^{\prime}}$,
we can construct an $(N^{\prime}, P^{\prime}, d, \varepsilon^{\prime})$-feedback code satisfying equal power constraint almost surely, namely,
{\small \begin{equation}\label{epec}
	\mathbb{P}\left[ \sum_{i=1}^{N^{\prime}} X_i^2 = N^{\prime} P^{\prime} \right] = 1.
\end{equation}}%
This is known as the Yaglom map trick \cite{FBL, Markov}, which implies that
{\small \begin{equation}
	N_{\text{fb, pp}}^{*}(P^{\prime}, d, \varepsilon^{\prime}) \geq N_{\text{fb, ep}}^{*}(P^{\prime}, d, \varepsilon^{\prime}) - 1,
\end{equation}}%
where $N_{\text{fb, ep}}^{*}(P^{\prime}, d, \varepsilon^{\prime})$ is the minimum coding blocklength for an $(N^{\prime}, P^{\prime}, d, \varepsilon^{\prime})$-feedback code under equal power constraint $P^{\prime}$, which completes
the proof of step $(c)$.

\subsubsection{\textbf{Proof of step (d)}}\label{sec4-1-4}
To complete the proof, we need to obtain $N_{\text{fb, ep}}^{*}\left( P^{\prime}, d, \varepsilon^{\prime} \right)$, which
is equivalent to finding a lower bound on $N^{\prime}$, and this depends on the following Lemma \ref{l2}.

\begin{lemma}\label{l2}
	Consider an $(N, P, d, \varepsilon)$-feedback code satisfying the equal power constraint {\small $\mathbb{P}\left[\sum_{i=1}^N X_i^2 = NP\right]=1$}.
	Let $q(Y^{N})=\prod_{i=1}^{N} q(Y_{i})$ be any product distribution on the output space,
	then for $\gamma >0$, we have
{\small \begin{equation}\label{jsk-ixiyi-vare}
		\begin{split}
			\mathbb{P}\left[\jmath_S(S,d) - \sum_{i=1}^{N}\imath(X_{i};Y_{i}) \geq \gamma \right] \leq \varepsilon +\exp\left(- \gamma\right),
		\end{split}
	\end{equation}}%
where $\imath(X_{i};Y_{i})$ is the information density defined as
{\small \begin{equation}\label{ixy}
	\imath(X_{i};Y_{i}) = \log\left(\frac{f(Y_{i}|X_{i})}{q(Y_{i})}\right),
\end{equation}}%
$f(\cdot)$ denotes probability density function,
$\jmath_S(S,d)$ is the $d$-tilted information in $S$ denoted by
{\small  \begin{equation}\label{jsd-d}
			\jmath_S(S, d) = -\ln \left(\mathbb{E}_{\widehat{S}^{*}}\left[\exp(\lambda^{*} d - \lambda^{*}(S-\widehat{S}^{*})^{2})\right]\right),
        \end{equation}}%
$\widehat{S}^{*}$ is the reproduction random variable that achieves the infimum in the rate-distortion function
{\small \begin{equation}\label{eq:rate_distortion}
			\mathbb{R}_{S}(d) = \inf_{\substack{P_{\widehat{S}|S}:\,\, \mathbb{E}[(S-\widehat{S})^{2}] \leq d}} I(S;\widehat{S}),
		\end{equation}}%
and $\lambda^*$ is the negative derivative of the rate-distortion function $\mathbb{R}_S(d)$ evaluated at $d$, namely, $\lambda^* = -\mathbb{R}_S^{\prime}(d)$.

\end{lemma}
\begin{IEEEproof}
	The proof of Lemma \ref{l2} is given in Appendix \ref{App-lemma1}.
\end{IEEEproof}

Applying Lemma \ref{l2} with $N$ substituted by $N^{\prime}$, $\varepsilon$ by $\varepsilon^{\prime}$, $P$ by $P^{\prime}$,
a lower bound on $N^{\prime}$ can be obtained via the following three steps.
{\small \begin{equation}\label{c-goal}
	\begin{split}
		&\mathbb{P}\left[\jmath_S(S,d) - \sum_{i=1}^{N^{\prime}}\imath(X_{i};Y_{i}) \geq \gamma \right] \leq \varepsilon^{\prime} +\exp\left(- \gamma\right)\\
		&\overset{(a^{*})}{\Rightarrow } \mathbb{P}\left[ U -\frac{1}{2}-\frac{N^{\prime}P^{\prime}}{2(P^{\prime}+\sigma_{\eta}^{2})} \geq \gamma- R(d)+N^{\prime}\cdot C(P^{\prime})\right]\\
		&\qquad \leq \varepsilon^{\prime} +\exp\left(- \gamma\right)\\
		&\overset{(b^{*})}{\Rightarrow } N^{\prime}\cdot C(P^{\prime})-R(d) \geq -\gamma\\	
		&\qquad +\sqrt{ V_{d}+N^{\prime}\cdot V(P^{\prime})}\cdot Q^{-1}\left(\varepsilon^{\prime} +\exp(-\gamma)+\frac{B^{*}}{\sqrt{N^{\prime}+1}}\right)\\
		&\overset{(c^{*})}{\Rightarrow }R(d) \leq F_{1}(N^{\prime})
	\end{split}
\end{equation}}%
where
$B^{*}$ is a finite constant and is determined by (\ref{B-star-form}) in Appendix \ref{App-37-stepac},
{\small
\begin{equation}\label{U-f}
	\begin{split}
	&U \triangleq \frac{S^{2}}{2\sigma_{s}^{2}}- \frac{1}{2(P^{\prime}+\sigma_{\eta}^{2})}\left(-\frac{P^{\prime}}{\sigma_{\eta}^{2}}(Y_{i}-X_{i})^{2}+2X_{i}(Y_{i}-X_{i})\right),
\end{split}
\end{equation}}%
{\small
\begin{equation}\label{f-f1}
	\begin{split}
		&F_{1}(N^{\prime})=N^{\prime}\cdot C\left(\frac{P}{1 - \varepsilon }\right) 
		 +\sqrt{V_{d}+ N^{\prime}\cdot V\left(\frac{P}{1-\varepsilon}\right)}\cdot \sqrt{\ln(N^{\prime})}+O\left(\sqrt{N^{\prime}}\right)
	\end{split}
\end{equation}}%
and
the proof of $(a^{*})-(c^{*})$  is in Appendix \ref{App-37-stepac}.

Furthermore, it is not difficult to check $F_{1}(N^{\prime})$ is increasing while $N^{\prime}$ is increasing
and $\lim\limits_{N^{\prime}\to +\infty} F_{1}(N^{\prime})=+\infty$,
which implies that there always exists $N_{1}\geq 2$ such that $F_{1}(N_{1}-1)<R(d)$ and
$F_{1}(N_{1})\geq R(d)$,
which completes the proof.

\subsection{Proof of Theorem \ref{th02}}\label{sec4-2}

\subsubsection{Coding procedure}\label{sec4-2-1}
\paragraph{\textbf{Initialization}}\label{sec4-2-1-1}
At time instant $i=1$, Alice sends the first input signal $X_{1}=\sqrt{\frac{P}{\sigma_{s}^{2}}}S$.
Meanwhile, Bob obtains the channel output
$Y_{1}=X_{1}+\eta_{1}$,
and adopts MMSE to compute the first estimation of $S$ by
{\small \begin{equation}\label{1to1-SE1}
	\widehat{S}_{1}=\frac{{\rm E}[S Y_{1}]}{{\rm E}[Y_{1}^{2}]}Y_{1}=\frac{P}{P+\sigma_{\eta}^{2}}S+\frac{\sqrt{P\sigma_{s}^{2}}}
{P+\sigma_{\eta}^{2}}\eta_{1}.
\end{equation}}%
The first estimation error is given by
{\small \begin{equation}\label{1to1-e1}
	\epsilon_{1}\triangleq \widehat{S}_{1}-S=\frac{\sqrt{P\sigma_{s}^{2}}}{P+\sigma_{\eta}^{2}}\eta_{1}
	-\frac{\sigma_{\eta}^{2}}{P+\sigma_{\eta}^{2}}S,
\end{equation}}%
and the variance of $\epsilon_{1}$ is
{\small \begin{equation}\label{1to1-a1}
	\alpha_{1}\triangleq  {\rm Var}(\epsilon_{1}) =\frac{\sigma_{\eta}^{2}\sigma_{s}^{2}}{P+\sigma_{\eta}^{2}}.
\end{equation}}%

\paragraph{\textbf{Iteration}}\label{sec4-2-1-2}
The remainder of the encoding-decoding procedure is the same as that of the Section \ref{sec3}.

\subsubsection{Lower bound derivation}\label{sec4-2-2}
\paragraph{\textbf{The proof of $\mathcal{C}^{\rm f}_{\rm s}(d,\varepsilon,\delta) \geq (N_{2})^{-1}$}}\label{sec4-2-2-1}
\begin{lemma}\label{pu}
	From our scheme stated above,
	the upper bound on $\mathcal{P}_{d, N}$
	is given by
	{\small \begin{equation}\label{1to1-s-d}
	\begin{split}
	\mathcal{P}_{d, N} \leq 2Q\left(\exp\left(-R(d)+N\cdot C(P)\right)\right),\\
\end{split}
\end{equation}}%
\end{lemma}
\begin{IEEEproof}
	The proof of Lemma \ref{pu} is given in Appendix \ref{App-pu}.
\end{IEEEproof}
Subsequently, to ensure that  $\mathcal{P}_{d, N}\leq \varepsilon$, let
{\small \begin{equation}\label{Qve}
	\begin{split}
	2Q\left(\exp\left(-R(d)+N\cdot C(P)\right)\right)\leq \varepsilon.
\end{split}
\end{equation}}%
Then, from (\ref{Qve}), we obtain
{\small \begin{equation}\label{1to1-N2-1}
	\begin{split}
		 N\geq \left(R(d)+\ln\left( Q^{-1}\left(\frac{\epsilon}{2}\right) \right)\right)\cdot \left(C(P) \right)^{-1}.
\end{split}
\end{equation}}%
Let
{\small \begin{equation}\label{1to1-N2-f}
	\begin{split}
		 N_{2}= \left\lceil \left(R(d)+\ln\left( Q^{-1}\left(\frac{\epsilon}{2}\right) \right)\right)\cdot \left(C(P) \right)^{-1} \right\rceil, \\
\end{split}
\end{equation}}%
the proof of  $\mathcal{C}^{\rm f}_{\rm s}(d,\varepsilon,\delta) \geq (N_{2})^{-1}$ is completed.

\paragraph{\textbf{The proof of $\mathcal{C}^{\rm f}_{\rm s}(d,\varepsilon,\delta) \geq (N_{3})^{-1}$}}\label{sec4-2-2-2}

\begin{lemma}\label{Ln-l}
For the eavesdropper, the upper bound of the information leakage rate is given by
{\small \begin{equation}\label{1to1-Ln}
	\begin{aligned}
		L_{N}\leq \frac{1}{2N}\left(\frac{(P+\sigma_{\eta}^{2})(\sigma_{e}^{2}+\widetilde{\sigma}_{e}^{2})}{\sigma_{e}^{2}\widetilde{\sigma}_{e}^{2}}\right)\left(1-\kappa^{N}\right)
		= F_{2}(N),
	\end{aligned}
\end{equation}}%
where $\kappa \triangleq \frac{\sigma_{\eta}^{2}}{P+\sigma_{\eta}^{2}}$.
\end{lemma}
\begin{IEEEproof}
	The proof of Lemma \ref{Ln-l} is given in Appendix \ref{App-ln}.
\end{IEEEproof}

Subsequently, following from Lemma \ref{Ln-l} and letting $F_{2}(N)\leq \delta$,
$L_{N}\leq \delta$ is guaranteed.
It is not difficult to check that $F_{2}(N)$ is decreasing while $N$ is increasing
and $\lim\limits_{N\to +\infty} F_{2}(N)=0$,
which implies that there always exists a $N_{3}\geq 1$ such that $F_{2}(N_{3}-1)>\delta$ and
$F_{2}(N_{3})\leq \delta$. Then, we obtain  $\mathcal{C}^{\rm f}_{\rm s}(d,\varepsilon,\delta) \geq (N_{3})^{-1}$, which competes the proof.

\section{Conclusion and Future Work }\label{sec5}

This paper establishes bounds on the FBL capacity of the
AWGN wiretap channel with noise-free feedback in JSCC setup. In particular, we propose a new lower bound which is based on a slight modification of the classic SK scheme, and numerical example shows that this new lower bound may be better than the conventional one. Besides this, we provide an
upper bound, which can also viewed as a FBL converse for the same model without the consideration of secrecy. Numerical results show that this upper bound may not be tight, and how to derive a tighter bound involving the secrecy constraint will be our future work.

\appendices
\onecolumn

\newpage
\section{The Proof of $(b^{\prime})$ in (\ref{edp})}\label{App-28-edpu}
\renewcommand{\theequation}{A\arabic{equation}}
\setcounter{equation}{0}
The excess-distortion probability $\mathcal{P}_{d, N}$ can be upper bounded as follows,
{\small \begin{align}
		&\mathcal{P}_{d, N}=\mathbb{P}\left[(S - \widehat{S}_{N})^{2}> d \right] \notag \\
		&= \mathbb{P}\left[\left\{ (S - \widehat{S}_{N})^{2}> d\right\} \cap \left\{\sum_{k=1}^{N} X_{k}^{2} > NP^{\prime} \right\} \right] \notag \\
		& + \mathbb{P}\left[ \left\{ (S - \widehat{S}_{N})^{2}> d\right\} \cap \left\{ \sum_{k=1}^{N} X_{k}^{2} \leq NP^{\prime} \right\} \right] \notag \\
		&\overset{(a^{\prime})}{\leq} \mathbb{P}\left[ \sum_{k=1}^{N} X_{k}^{2} > NP^{\prime} \right] \notag \\
		& + \mathbb{P}\left[ \left\{ (S - \widehat{S}_{N})^{2}> d\right\} \cap \left\{ \sum_{k=1}^{N} X_{k}^{2} \leq NP^{\prime} \right\} \right] \notag \\
		&\overset{(b^{\prime})}{=} \mathbb{P} \left[ \sum_{k=1}^{N} X_{k}^{2} > NP^{\prime} \right]\notag \\
		& +\mathbb{P} \left[ \left\{ (S - \widehat{S}_{N})^{2}> d \right\} \cap \{ X_k' = X_k, k = 1, \ldots, N \} \right] \notag \\
		&\overset{(c^{\prime})}{\leq} \varepsilon +  \mathbb{P} \left[ \sum_{k=1}^{N} X_{k}^{2} > NP^{\prime} \right]\notag\\
		& \overset{(d^{\prime})}{\leq} \varepsilon + \frac{ \sum\limits_{k=1}^{N} \mathbb{E}[X_k^2]}{ NP^{\prime}}\notag \\
		&\overset{(e^{\prime})}{\leq} \varepsilon + \frac{NP}{NP^{\prime}} \notag \\
		&\overset{(f^{\prime})}{\leq} \varepsilon + 1 - \varepsilon - \frac{\zeta}{\sqrt{N}}  \notag \\
		&= 1 - \frac{\zeta}{\sqrt{N}} \,\,\triangleq\,\, \varepsilon^{\prime},
\end{align}}%
where\\
$(a^{\prime})$ follows from $\mathbb{P}\left[A\bigcap B\right] \leq\mathbb{P}\left[ B\right]$;\\
$(b^{\prime})$ follows from (\ref{Xprime});\\
$(c^{\prime})$ follows from the excess-distortion probability $\mathcal{P}_{d, N}$ of the original code is upper bound by $\varepsilon$;\\
$(d^{\prime})$ follows from Markov's inequality \cite{Markov};\\
$(e^{\prime})$ follows from the average power constraint of the original code;\\
$(f^{\prime})$ follows from $P^{\prime} \triangleq \frac{P}{1 - \varepsilon - \frac{\zeta}{\sqrt{N}}}$.

\section{Proof of Lemma \ref{l2} in Section \ref{sec4-1-4}}\label{App-lemma1}
\renewcommand{\theequation}{B\arabic{equation}}
\setcounter{equation}{0}

{\small \begin{align}\label{jsk-ixiyi-fl}
		 &\mathbb{P}\left[\jmath_S(S,d) - \sum_{i=1}^{N}\imath(X_{i};Y_{i}) \geq \gamma \right] \notag \\
		 &=\mathbb{P}\left[\left(\jmath_S(S,d) - \sum_{i=1}^{N}\imath(X_{i};Y_{i}) \geq \gamma \right) \bigcap \left((S-\widehat{S})^{2} \geq d \right)\right]
		 +\mathbb{P}\left[\left(\jmath_S(S,d) - \sum_{i=1}^{N}\imath(X_{i};Y_{i}) \geq \gamma \right) \bigcap \left((S-\widehat{S})^{2} < d \right)\right]\notag \\
		 &\overset{(a)}{\leq } \mathbb{P}\left[\left((S-\widehat{S})^{2} \geq d \right)\right]
		 +\mathbb{P}\left[\left(\jmath_S(S,d) - \sum_{i=1}^{N}\imath(X_{i};Y_{i}) \geq \gamma \right) \bigcap \left((S-\widehat{S})^{2} < d \right)\right]\notag \\
		 &\overset{(b)}{\leq } \varepsilon +\mathbb{P}\left[\left(\jmath_S(S,d) - \sum_{i=1}^{N}\imath(X_{i};Y_{i}) \geq \gamma \right) \bigcap \left((S-\widehat{S})^{2} < d \right)\right]\notag \\
		 &\overset{(c)}{=}\varepsilon + \mathbb{E}_{S,X^{N},Y^{N},\widehat{S}}  \left[\indicator{\left(\jmath_S(S,d) - \sum_{i=1}^{N}\imath(X_{i};Y_{i}) \geq \gamma \right) \bigcap \left((S-\widehat{S})^{2} < d \right) }\right]\notag \\
		 &= \varepsilon + \int_{s,x^{N},y^{N},\widehat{s}} f(s,x^{N},y^{N},\widehat{s})\cdot \indicator{\left(\jmath_S(s,d) - \sum_{i=1}^{N}\imath(x_{i};y_{i}) \geq \gamma \right) \bigcap \left((s-\widehat{s})^{2} < d \right) }dsdx^{N}dy^{N}d\widehat{s}\notag \\
		 &=\varepsilon + \int_{s,x^{N},y^{N},\widehat{s}} f(s)\cdot f(x^{N},y^{N}|s)\cdot f(\widehat{s}|s,x^{N},y^{N}) \cdot \indicator{\left(\jmath_S(s,d) - \sum_{i=1}^{N}\imath(x_{i};y_{i}) \geq \gamma \right) \bigcap \left((s-\widehat{s})^{2} < d \right) }dsdx^{N}dy^{N}d\widehat{s}\notag \\
		 &\overset{(d)}{=}\varepsilon +\int_{s,x^{N},y^{N},\widehat{s}}  f(s)\cdot \left(\prod_{i=1}^{N} f(x_{i},y_{i}|x^{i-1},y^{i-1},s)\right) \cdot f(\widehat{s}|y^{N}) \cdot \indicator{\left(\jmath_S(s,d) - \sum_{i=1}^{N}\imath(x_{i};y_{i}) \geq \gamma \right) \bigcap \left((s-\widehat{s})^{2} < d \right)}dsdx^{N}dy^{N}d\widehat{s}\notag \\
		 &\overset{(e)}{=}\varepsilon +\int_{s,x^{N},y^{N}}  f(s)\cdot \left(\prod_{i=1}^{N} f(x_{i},y_{i}|x^{i-1},y^{i-1},s)\right) \cdot \left(\int_{\widehat{s}\in B_{d}(s)}f(\widehat{s}|y^{N})d\widehat{s}\right) \cdot \indicator{\jmath_S(s,d) - \sum_{i=1}^{N}\imath(x_{i};y_{i}) \geq \gamma }dsdx^{N}dy^{N}\notag \\		
		 &=\varepsilon +\int_{s,x^{N},y^{N}} f(s)\cdot \left(\prod_{i=1}^{N} \left(f(x_{i}|x^{i-1},y^{i-1},s) \cdot  f(y_{i}|x^{i},y^{i-1},s)\right) \right)\cdot \left(\int_{\widehat{s}\in B_{d}(s)}f(\widehat{s}|y^{N})d\widehat{s}\right)\cdot \indicator{\jmath_S(s,d) - \sum_{i=1}^{N}\imath(x_{i};y_{i}) \geq \gamma }dsdx^{N}dy^{N}\notag \\		
		 &\overset{(f)}{=} \varepsilon +\int_{s,x^{N},y^{N}} f(s)\cdot \left(\prod_{i=1}^{N} \left(f(x_{i}|x^{i-1},y^{i-1},s) \cdot f(y_{i}|x_{i}) \right) \right)\cdot \left(\int_{\widehat{s}\in B_{d}(s)}f(\widehat{s}|y^{N})d\widehat{s}\right)\cdot \indicator{\jmath_S(s,d) - \sum_{i=1}^{N}\imath(x_{i};y_{i}) \geq \gamma }dsdx^{N}dy^{N}\notag \\		 	
		 &\overset{(g)}{\leq} \varepsilon +\int_{s,x^{N},y^{N}}f(s)\cdot \left(\int_{\widehat{s}\in B_{d}(s)}f(\widehat{s}|y^{N})d\widehat{s}\right) \cdot \left(\prod_{i=1}^{N}q(y_{i})\cdot f(x_{i}|x^{i-1},y^{i-1},s) \right) \cdot \exp\left(\jmath_S(s,d) - \gamma\right)dsdx^{N}dy^{N} \notag \\	
		 &\overset{(h)}{=} \varepsilon +\int_{s,y^{N}}f(s)\cdot \left(\int_{\widehat{s}\in B_{d}(s)}f(\widehat{s}|y^{N})d\widehat{s}\right) \cdot q(y^{N})\cdot \left(\int_{x^{N}}\prod_{i=1}^{N}  f(x_{i}|x^{i-1},y^{i-1},s) dx^{N}\right)\cdot \exp\left(\jmath_S(s,d) - \gamma\right)dsdy^{N}\notag \\
		 &= \varepsilon +\int_{s,y^{N}}f(s)\cdot \left(\int_{\widehat{s}\in B_{d}(s)}f(\widehat{s}|y^{N})d\widehat{s}\right) \cdot q(y^{N}) \cdot 1 \cdot \exp\left(\jmath_S(s,d) - \gamma\right)dsdy^{N}\notag \\
		 &\overset{(i)}{=}\varepsilon +\int_{s}f(s)\cdot \left(\int_{\widehat{s}\in B_{d}(s)}\int_{y^{N}}f^{*}(\widehat{s},y^{N})dy^{N}d\widehat{s}\right) \cdot \exp\left(\jmath_S(s,d) - \gamma\right)ds\notag \\
		 &= \varepsilon +\int_{s}f(s)\cdot \left(\int_{\widehat{s}\in B_{d}(s)}f^{*}(\widehat{s})d\widehat{s}\right)  \cdot \exp\left(\jmath_S(s,d) - \gamma\right)ds\notag \\
		 &\overset{(j)}{\leq} \varepsilon +\int_{s}f(s)\cdot \left(\int_{\widehat{s}\in B_{d}(s)}f^{*}(\widehat{s})\cdot \exp\left(- \gamma\right) \cdot \exp\left(\jmath_S(s,d) +\lambda^{*}d-\lambda^{*}(s-\widehat{s})^{2}\right)d\widehat{s}\right)ds\notag \\
		 &\overset{(k)}{\leq} \varepsilon +\int_{s}f(s)\cdot \left(\int_{\widehat{s}\in B_{d}(s)}f^{*}(\widehat{s})\cdot \exp\left(- \gamma\right) \cdot \exp\left(\jmath_S(s,d) +\lambda^{*}d-\lambda^{*}(s-\widehat{s})^{2}\right)d\widehat{s}\right)ds\notag \\
		 &\,\,+\int_{s}f(s)\cdot \left(\int_{\widehat{s}\in \overline{B_{d}(s)} }f^{*}(\widehat{s}) \cdot \exp\left(- \gamma\right) \cdot \exp\left(\jmath_S(s,d) +\lambda^{*}d-\lambda^{*}(s-\widehat{s})^{2}\right)d\widehat{s}\right)ds\notag \\
		 &= \varepsilon +\int_{s}f(s)\cdot \left(\int_{\widehat{s}}f^{*}(\widehat{s})\cdot \exp\left(- \gamma\right) \cdot \exp\left(\jmath_S(s,d) +\lambda^{*}d-\lambda^{*}(s-\widehat{s})^{2}\right)d\widehat{s}\right)ds\notag \\
		 &= \varepsilon +  \exp\left(- \gamma\right) \cdot \left( \int_{\widehat{s}}f^{*}(\widehat{s})\int_{s}f(s)\cdot \exp\left(\jmath_S(s,d) +\lambda^{*}d-\lambda^{*}(s-\widehat{s})^{2}\right)dsd\widehat{s}\right)\notag \\		
		 &= \varepsilon +  \exp\left(- \gamma\right) \cdot \left( \int_{\widehat{s}}f^{*}(\widehat{s})\cdot \mathbb{E}\left[\exp\left(\jmath_S(S,d) +\lambda^{*}d-\lambda^{*}(S-\widehat{s})^{2}\right)\right]d\widehat{s}\right) \notag \\		
		 &\overset{(l)}{\leq}\varepsilon +\exp\left(- \gamma\right)\cdot \int_{\widehat{s}}f^{*}(\widehat{s}) d\widehat{s}   \notag \\
		 &=\varepsilon +\exp\left(- \gamma\right),
\end{align}}%
where $x^{N}\triangleq(x_{1},\ldots,x_{N})$, $dx^{N}\triangleq dx_{1}\cdot dx_{2} \ldots \cdot dx_{N}$,
$y^{N}\triangleq(y_{1},\ldots,y_{N})$, $dy^{N}\triangleq dy_{1}\cdot dy_{2}\ldots \cdot dy_{N}$,
$\indicator{\cdot}$ is the indicator function,\\
$(a)$ follows from $\mathbb{P}\left[A\bigcap B\right] \leq\mathbb{P}\left[ B\right]$;\\
$(b)$ follows from $\mathbb{P}\left[ (S-\widehat{S})^{2} \geq d \right] \leq \varepsilon$;\\
$(c)$ follows from $\mathbb{P}\left[\cdot\right]=\mathbb{E}\left[\indicator{\cdot}\right]$;\\
$(d)$ follows from the chain rule of probability density function and the Markov chain $(S, X^N) \to Y^N \to \widehat{S}$;\\
$(e)$ follows from $B_{d}(s) = \{ \widehat{s} \in \widehat{\mathcal{S}}:(s-\widehat{s})^{2} < d \}$ and $\widehat{\mathcal{S}}$ denotes all possible estimatations of $S$;\\
$(f)$ follows from the channel is memoryless;\\
$(g)$ follows from
\begin{equation}\label{jsk-ixiyi}
	\begin{split}
		&\jmath_S(s,d) - \sum_{i=1}^{N}\imath(x_{i};y_{i}) \geq \gamma\\
		&\Rightarrow \jmath_S(s,d) - \sum_{i=1}^{N}\log(\frac{f(y_{i}|x_{i})}{q(y_{i})}) \geq \gamma\\
		&\Rightarrow \jmath_S(s,d) - \log(\prod_{i=1}^{N}\frac{f(y_{i}|x_{i})}{q(y_{i})}) \geq \gamma\\
		&\Rightarrow \exp\left(\jmath_S(s,d) - \gamma\right)\geq \prod_{i=1}^{N}\frac{f(y_{i}|x_{i})\cdot f(x_{i}|x^{i-1},y^{i-1},s)}{q(y_{i})\cdot f(x_{i}|x^{i-1},y^{i-1},s)} \\
		&\Rightarrow \left(\prod_{i=1}^{N}q(y_{i})\cdot f(x_{i}|x^{i-1},y^{i-1},s) \right) \cdot \exp\left(\jmath_S(s,d) - \gamma\right)\geq \prod_{i=1}^{N}f(y_{i}|x_{i})\cdot f(x_{i}|x^{i-1},y^{i-1},s) ;\\
	\end{split}
\end{equation}
$(h)$ follows from $q(y^{N})=\prod_{i=1}^{N}q(y_{i})$;\\
$(i)$ follows from $f^{*}(\widehat{s},y^{N})\triangleq f(\widehat{s}|y^{N})\cdot q(y^{N})$;\\
$(j)$ follows from $(s-\widehat{s})^{2} < d$;\\
$(k)$ follows from that
{\small \begin{equation}\label{bdsbuji}
	\begin{split}
 	\int_{s}f(s)\cdot \left(\int_{\widehat{s}\in \overline{B_{d}(s)} }f^{*}(\widehat{s}) \cdot \exp\left(- \gamma\right) \cdot \exp\left(\jmath_S(s,d) +\lambda^{*}d-\lambda^{*}(s-\widehat{s})^{2}\right)d\widehat{s}\right)ds\geq 0,
	\end{split}
\end{equation}}%
$\overline{B_{d}(s)}$ is the complement of $B_{d}(s)$ and $B_{d}(s) = \{ \widehat{s} \in \widehat{\mathcal{S}} : (s-\widehat{s})^{2} < d \}$;\\
$(l)$ follows from similar argument in \cite[equation (17) in P.3]{FBL-JSCC}  , where $\mathbb{E}\left[ \exp(\lambda^*d - \lambda^*(S-\widehat{s})^{2} + \jmath_S(S, d)) \right] \leq 1$, for all $\widehat{s}\in \widehat{\mathcal{S}}$.\\

\section{The Proof of steps $(a^{*})-(c^{*})$ in (\ref{c-goal})}\label{App-37-stepac}
\renewcommand{\theequation}{C\arabic{equation}}
\setcounter{equation}{0}
\subsection{Proof of step $(a^{*})$ in (\ref{c-goal})}\label{App2-1}
Applying Lemma \ref{l2} with $N$ substituted by $N^{\prime}$, $\varepsilon$ by $\varepsilon^{\prime}$, $P$ by $P^{\prime}$,
we have
{\small \begin{equation}\label{c-goal-fl}
	\begin{split}
		&\mathbb{P}\left[\jmath_S(S,d) - \sum_{i=1}^{N^{\prime}}\imath(X_{i};Y_{i}) \geq \gamma \right] \leq \varepsilon^{\prime} +\exp\left(- \gamma\right).
	\end{split}
\end{equation}}%
where $\imath(X_{i};Y_{i})$ is the information density defined as
{\small \begin{equation}\label{ixy-fl}
	\imath(X_{i};Y_{i}) = \log\left(\frac{f(Y_{i}|X_{i})}{q(Y_{i})}\right),
\end{equation}}%
$f(\cdot)$ denotes probability density function,
$\jmath_S(S,d)$ is the $d$-tilted information in $S$ denoted by
{\small  \begin{equation}\label{jsd-d-fl}
			\jmath_S(S, d) = -\ln \left(\mathbb{E}_{\widehat{S}^{*}}\left[\exp(\lambda^{*} d - \lambda^{*}(S-\widehat{S}^{*})^{2})\right]\right),
        \end{equation}}
$\widehat{S}^{*}$ is the reproduction random variable that achieves the infimum in the rate-distortion function
{\small \begin{equation}\label{eq:rate_distortion}
			\mathbb{R}_{S}(d) = \inf_{\substack{P_{\widehat{S}|S}:\,\, \mathbb{E}[(S-\widehat{S})^{2}] \leq d}} I(S;\widehat{S}),
		\end{equation}}%
and $\lambda^*$ is the negative derivative of the rate-distortion function $\mathbb{R}_S(d)$ evaluated at $d$, namely, $\lambda^* = -\mathbb{R}_S^{\prime}(d)$.

\subsubsection{The calculation of $\jmath_S(S,d)$}\label{App2-1-1}
For a Gaussian source $S\sim\mathcal{N}(0,\sigma_s^2)$, the rate-distortion function and the parameters of the optimal test channel are well known\cite{EIT},
\begin{equation}\label{AWGN-Rsd}
\begin{split}
    & \mathbb{R}_{S}(d)=\frac{1}{2}\ln\left(\frac{\sigma_{s}^{2}}{d}\right)\triangleq R(d),
	\qquad \widehat{S}^{*} \sim \mathcal{N}(0,\,\sigma_{s}^{2}-d).
\end{split}
\end{equation}
Subsequently, since $\lambda^* = -\mathbb{R}_{S}^{\prime}(d) = \frac{1}{2d}$,
the expectation appearing in (\ref{jsd-d-fl}) is re-written by
{\small \begin{equation}\label{EJs1}
\begin{split}
    \mathbb{E}_{\widehat{S}^{*}}\left[\exp\left(\lambda^* d - \lambda^* (S-\widehat{S}^{*})^{2}\right)\right]
	& = \mathbb{E}_{\widehat{S}^{*}}\left[\exp\left(\frac{1}{2} - \frac{(S - \widehat{S}^{*})^{2}}{2d}\right)\right]
	= \exp\left(\frac{1}{2}\right) \cdot \mathbb{E}_{\widehat{S}^{*}}\left[\exp\left(-\frac{(S - \widehat{S}^{*})^{2}}{2d}\right)\right]\\
    &\overset{(a)}{=}\exp\left(\frac{1}{2}\right) \cdot\exp\left(-\frac{S^{2}}{2\sigma_{s}^{2}}\right) \cdot \sqrt{\frac{d}{\sigma_{s}^{2}}},
\end{split}
\end{equation}}%
where $(a)$ follows from
{\small \begin{equation}\label{EJs2}
\begin{split}
    &\mathbb{E}_{\widehat{S}^{*}}\left[\exp\left(- \frac{(S - \widehat{S}^{*})^{2}}{2d}\right)\right] = \int_{-\infty}^{+\infty} \frac{1}{\sqrt{2\pi (\sigma_{s}^{2}-d)}}
     \cdot \exp\left(-\frac{(\widehat{s}^{*})^{2}}{2(\sigma_{s}^{2}-d)}\right) \cdot \exp\left(-\frac{(S-\widehat{s}^{*})^{2}}{2d}\right)\,d\widehat{s}^{*} \\
    &= \exp\left(-\frac{S^{2}}{2\sigma_{s}^{2}}\right)
	\cdot\frac{1}{\sqrt{2\pi (\sigma_{s}^{2}-d)}}
	\cdot\sqrt{\frac{2\pi (\sigma_{s}^{2}-d)d}{\sigma_{s}^{2}}}
	\cdot \int_{-\infty}^{+\infty} \sqrt{\frac{\sigma_{s}^{2}}{2\pi (\sigma_{s}^{2}-d)d}}
	\cdot \exp\left(-\frac{\sigma_{s}^{2}}{2d(\sigma_{s}^{2}-d)} \cdot \left(\widehat{s}^{*} - \frac{S(\sigma_{s}^{2}-d)}{\sigma_{s}^{2}}\right)^2\right)\,d\widehat{s}^{*} \\
    &= \exp\left(-\frac{S^{2}}{2\sigma_{s}^{2}}\right) \cdot \sqrt{\frac{d}{\sigma_{s}^{2}}}.
\end{split}
\end{equation}}

Finally, substituting (\ref{EJs1}) into (\ref{jsd-d-fl}), we have
\begin{equation}\label{jsd}
    \jmath_S(S, d) = R(d) - \frac{1}{2} + \frac{S^{2}}{2\sigma_{s}^{2}}.
\end{equation}
\subsubsection{The calculation of $\sum_{i=1}^{N^{\prime}}\imath(X_{i};Y_{i})$}\label{App2-1-2}
On the other hand, for the AWGN channel $Y_i = X_i + \eta_i$ with $\eta_i \sim \mathcal{N}(0, \sigma_\eta^2)$,
we choose the spherically symmetric $q(Y^{N^{\prime}})=\prod_{i=1}^{N^{\prime}} q(Y_{i})$,
where $q(Y_i) \sim \mathcal{N}(0, P^{\prime}+\sigma_{\eta}^{2})$
is the capacity achieving output distribution.
Then, for each $i$,
{\small \begin{equation}\label{ixiyi}
	\begin{split}
		\imath(X_{i};Y_{i}) &=\ln\left(\frac{f(Y_{i}|X_{i})}{q(Y_{i})}\right)
		\overset{(b)}{=}\ln\left(\sqrt{1+\frac{P^{\prime}}{\sigma_{\eta}^{2}}}\cdot \exp\left(\frac{Y_{i}^{2}}{2(P^{\prime}+\sigma_{\eta}^{2})}-\frac{(Y_{i}-X_{i})^{2}}{2\sigma_{\eta}^{2}}\right)\right)\\
		&=\frac{1}{2}\ln\left(1+\frac{P^{\prime}}{\sigma_{\eta}^{2}}\right)+\frac{1}{2(P^{\prime}+\sigma_{\eta}^{2})}\left(Y_{i}^{2}-(1+\frac{P^{\prime}}{\sigma_{\eta}^{2}})(Y_{i}-X_{i})^{2}\right)\\
        &=\frac{1}{2}\ln\left(1+\frac{P^{\prime}}{\sigma_{\eta}^{2}}\right)
		+\frac{1}{2(P^{\prime}+\sigma_{\eta}^{2})}\left(-\frac{P^{\prime}}{\sigma_{\eta}^{2}}(Y_{i}-X_{i})^{2}+X_{i}^{2}+2X_{i}(Y_{i}-X_{i})\right),
	\end{split}
\end{equation}}%
where $(b)$ follows from
{\small \begin{equation}\label{fxy-qy}
	\begin{split}
		&f(Y_{i}|X_{i})=\frac{1}{\sqrt{2\pi \sigma_{\eta}^{2}}}\exp\left(-\frac{(Y_{i}-X_{i})^{2}}{2\sigma_{\eta}^{2}}\right),
        \quad q(Y_{i})=\frac{1}{\sqrt{2\pi (P^{\prime}+\sigma_{\eta}^{2})}}\exp\left(-\frac{Y_{i}^{2}}{2(P^{\prime}+\sigma_{\eta}^{2})}\right).\\
\end{split}
\end{equation}}%
Then, we have
{\small \begin{equation}\label{sumixiyi}
	\begin{split}
      \sum_{i=1}^{N^{\prime}}\imath(X_{i};Y_{i}) =
	  N^{\prime}\cdot C(P^{\prime})+\frac{1}{2(P^{\prime}+\sigma_{\eta}^{2})}\cdot \sum_{i=1}^{N^{\prime}} X_i^2
		+\sum_{i=1}^{N^{\prime}}\Psi(X_{i}, Y_{i})
	\end{split}
\end{equation}}%
where
{\small \begin{equation}\label{lambdaxy}
	\Psi(X_{i}, Y_{i}) \triangleq \frac{1}{2(P^{\prime}+\sigma_{\eta}^{2})}\left(-\frac{P^{\prime}}{\sigma_{\eta}^{2}}(Y_{i}-X_{i})^{2}+2X_{i}(Y_{i}-X_{i})\right).
\end{equation}}%
Furthermore, under the equal power constraint ( $\mathbb{P}\left[\sum_{i=1}^{N^{\prime}} X_i^2 = N^{\prime}P^{\prime}\right]=1$), it is not difficult to check that
{\small \begin{equation}\label{prsumixiyi}
	\begin{split}
      \mathbb{P}\left[\sum_{i=1}^{N^{\prime}}\imath(X_{i};Y_{i}) = N^{\prime}\cdot C(P^{\prime})+\frac{N^{\prime}P^{\prime}}{2(P^{\prime}+\sigma_{\eta}^{2})}
		+\sum_{i=1}^{N^{\prime}}\Psi(X_{i}, Y_{i})\right]	 =1
	\end{split}
\end{equation}}%
where the Gaussian capacity function is defined as $C(x)\triangleq \frac{1}{2}\ln\left(1+\frac{x}{\sigma_{\eta}^{2}}\right)$.

\begin{lemma}\label{equalp}
Let $X^{*}$ and $Y^{*}$ be random variables defined on the same probability space.
If $\mathbb{P}(X^{*} = Y^{*}) = 1$, then for any $r \in \mathbb{R}$,
\begin{equation}\label{XY*}
\begin{split}
\mathbb{P}[X^{*} \leq r] = \mathbb{P}[Y^{*} \leq r].
\end{split}
\end{equation}
\end{lemma}

\begin{proof}
The proof of Lemma \ref{equalp} is given in Appendix \ref{App-equalp}.
\end{proof}
From Lemma \ref{equalp}, we have
{\small \begin{equation}\label{prsumixiyi-2}
	\begin{split}
      &\mathbb{P}\left[\jmath_S(S,d) - \sum_{i=1}^{N^{\prime}}\imath(X_{i};Y_{i}) \geq \gamma \right]
	  =\mathbb{P}\left[\jmath_S(S,d) - \left(N^{\prime}\cdot C(P^{\prime})+\frac{N^{\prime}P^{\prime}}{2(P^{\prime}+\sigma_{\eta}^{2})}
		+\sum_{i=1}^{N^{\prime}}\Psi(X_{i}, Y_{i}) \right) \geq \gamma \right]
	\end{split}
\end{equation}}%
Finally, substituting (\ref{jsd}) into (\ref{prsumixiyi-2}), we have
{\small \begin{equation}\label{c-goal-a*}
	\mathbb{P}\left[ U -\frac{1}{2}-\frac{N^{\prime}P^{\prime}}{2(P^{\prime}+\sigma_{\eta}^{2})} \geq \gamma- R(d)+N^{\prime}\cdot C(P^{\prime})\right]\leq \varepsilon^{\prime} +\exp\left(- \gamma\right),
\end{equation}}%
where
\begin{equation}\label{U}
	U \triangleq \frac{S^{2}}{2\sigma_{s}^{2}}-\sum_{i=1}^{N^{\prime}}\Psi(X_{i}, Y_{i}).
\end{equation}%

\subsection{Proof of step $(b^{*})$ in (\ref{c-goal})}\label{App2-2}

The following Lemma \ref{MGF-equality} is the key to prove step $(b^{*})$.

\begin{lemma}\label{MGF-equality}
Let $U$ be defined in (\ref{U}), and $T$ be a random variable given by
\begin{equation}\label{T}
	T\triangleq \frac{1}{2} G^{2}+\sum_{i=1}^{N^{\prime}}\left(\frac{P^{\prime}}{2(P^{\prime}+\sigma_{\eta}^{2})}K_{i}^{2}-\frac{\sqrt{P^{\prime}\sigma_{\eta}^{2}}}{P^{\prime}+\sigma_{\eta}^{2}}K_{i}\right),
\end{equation}
where $G$ and $K_{i} (i \in \{1,\ldots, N^{\prime}\})$ are independent normalized Gaussian random variables.
Then, the moment generating functions of $U$ and $T$ are equal, namely,
\begin{equation}
    \mathbb{E}_{S, X^{N^{\prime}}, Y^{N^{\prime}}} \left[ \exp(tU) \right] = \mathbb{E}_{G,K^{N^{\prime}}}\left[ \exp(tT) \right],
\end{equation}
where $X^{N^{\prime}}\triangleq(X_{1},\ldots,X_{N^{\prime}})$,
$Y^{N^{\prime}}\triangleq(Y_{1},\ldots,Y_{N^{\prime}})$
and $K^{N^{\prime}}\triangleq(K_{1},\ldots,K_{N^{\prime}})$.
\end{lemma}
\begin{IEEEproof}
	The proof of Lemma \ref{MGF-equality} is given in Appendix \ref{App-MGFE}.
\end{IEEEproof}
From Lemma \ref{MGF-equality} and similar argument in \cite{FBL-feedback, jumu}, we conclude that
\begin{equation}\label{T-U}
    \begin{split}
    &\mathbb{P}\left[ U -\frac{1}{2}-\frac{N^{\prime}P^{\prime}}{2(P^{\prime}+\sigma_{\eta}^{2})} \geq \gamma- R(d)+N^{\prime}\cdot C(P^{\prime})\right]
    =\mathbb{P}\left[ T -\frac{1}{2}-\frac{N^{\prime}P^{\prime}}{2(P^{\prime}+\sigma_{\eta}^{2})} \geq \gamma- R(d)+N^{\prime}\cdot C(P^{\prime})\right].
\end{split}
\end{equation}
Furthermore, define
{\begin{equation}
    \begin{split}
    &T^{*}_{0} \triangleq \frac{1}{2}(G^2-1),
    \quad T^{*}_{i} \triangleq \frac{P^{\prime}}{2(P^{\prime}+\sigma_{\eta}^{2})}(K_i^2-1) - \frac{\sqrt{P^{\prime}\sigma_{\eta}^{2}}}{P^{\prime}+\sigma_{\eta}^{2}}K_i, \quad 1\leq i \leq N^{\prime}.
    \end{split}
\end{equation}}%
Then $T-\frac{1}{2}-\frac{N^{\prime}P^{\prime}}{2(P^{\prime}+\sigma_{\eta}^{2})}=\sum_{i=0}^{N^{\prime}} T^{*}_{i}$, where $T^{*}_{i}$ is independent of each other.

\begin{lemma}\label{E-Var-jsj-T}
    Let $T^{*}_{i} (i\in \{0,\ldots,N^{\prime}\})$ be defined as above.
    Then,
\begin{equation}\label{sum-e-var-jsj}
    \begin{split}
    &\sum_{i=0}^{N^{\prime}}\mathbb{E}[ T^{*}_{i}]=0,
    \quad \sum_{i=0}^{N^{\prime}} {\rm Var}\left[  T^{*}_{i} \right] = V_{d}+N^{\prime}\cdot V(P^{\prime}),\\
    &\rho \triangleq \sum_{0}^{N^{\prime}}\mathbb{E}[|(T^{*}_i)^3|] =\sum_{0}^{N^{\prime}}\rho_{i}= \rho_0+ N^{\prime}\cdot \rho_{1},\\
	& 0 \leq \rho < +\infty.
    \end{split}
\end{equation}
where $V_{d} \triangleq \frac{1}{2}$, $V(x)\triangleq \frac{x(x+2\sigma_{\eta}^{2})}{2(x+\sigma_{\eta}^{2})^{2}}$ denotes the Gaussian dispersion function,
and for $i \in\{0,\ldots,N^{\prime}\}$, $\rho_i \triangleq \mathbb{E}[|(T^{*}_i)^3|]$ denotes the third absolute moment.
\end{lemma}
\begin{IEEEproof}
	The proof of Lemma \ref{E-Var-jsj-T} is given in Appendix \ref{App-EVJ}.
\end{IEEEproof}
Subsequently, from Lemma \ref{E-Var-jsj-T}
and applying the Berry-Esseen theorem \cite{berry-essen},
we have
\begin{equation}\label{jsk-ixiyi-awgn-berry-essen}
	\begin{split}
		&\left\lvert \mathbb{P}\left[ T -\frac{1}{2}-\frac{N^{\prime}P^{\prime}}{2(P^{\prime}+\sigma_{\eta}^{2})} \geq \gamma-R(d)+N^{\prime}C(P^{\prime})\right] - Q\left(\frac{\gamma-R(d)+N^{\prime}C(P^{\prime})}{\sqrt{ V_{d}+N^{\prime}\cdot V(P^{\prime})}}\right) \right\rvert
        \leq 6\frac{\rho}{(V_{d}+N^{\prime}\cdot V(P^{\prime}))^{\frac{3}{2}}}.
	\end{split}
\end{equation}
Using $\rho \leq \max(\rho_0,\rho_1)\cdot (N^{\prime}+1)$ and {\small $V_d+N^{\prime}\cdot V(P^{\prime})\geq \min(V_d,V(P^{\prime}))\cdot (N^{\prime}+1)$}, where $\max(x,y)=x$ and $\min(x,y)=y$  while $x\geq y$, we obtain
\begin{equation}\label{bound-simplify}
    \begin{split}
        6\frac{\rho}{(V_{d}+N^{\prime}\cdot V(P^{\prime}))^{\frac{3}{2}}}
        \leq 6 \frac{\max(\rho_0, \rho_1)}{[\min(V_d, V(P^{\prime}))]^{3/2}} \cdot \frac{1}{\sqrt{N^{\prime}+1}}
        = \frac{B^{*}}{\sqrt{N^{\prime}+1}},
    \end{split}
\end{equation}
where
\begin{equation}\label{B-star-form}
    B^{*}=6\frac{\max(\rho_0,\rho_1)}{\bigl[\min(V_d,V(P^{\prime}))\bigr]^{3/2}}>0.
\end{equation}
Moreover, since $\rho_0,\rho_1<+\infty$ and $\min(V_d,V(P^{\prime}))>0$ for $P^{\prime}>0$ and $\sigma_{\eta}^2>0$,
the constant $B^{*}$ is also finite.
Consequently,
{\small \begin{equation}\label{jsk-ixiyi-awgn-berry-essen-b}
	\begin{split}
		\left\lvert \mathbb{P}\left[ T -\frac{1}{2}-\frac{N^{\prime}P^{\prime}}{2(P^{\prime}+\sigma_{\eta}^{2})} \geq \gamma-R(d)+N^{\prime}\cdot C(P^{\prime})\right] \right.
		\left. - Q\left(\frac{\gamma-R(d)+N^{\prime}\cdot C(P^{\prime})}{\sqrt{ V_{d}+N^{\prime}\cdot V(P^{\prime})}}\right) \right\rvert
        \leq \frac{B^{*}}{\sqrt{N^{\prime}+1}},
	\end{split}
\end{equation}}%
which implies that
\begin{equation}\label{jsk-ixiyi-awgn-berry-essen-fw}
	\begin{split}
		&\mathbb{P}\left[ T -\frac{1}{2}-\frac{N^{\prime}P^{\prime}}{2(P^{\prime}+\sigma_{\eta}^{2})} \geq \gamma-R(d)+N^{\prime}\cdot C(P^{\prime})\right]
		\geq Q\left(\frac{\gamma-R(d)+N^{\prime}\cdot C(P^{\prime})}{\sqrt{V_{d}+N^{\prime}\cdot V(P^{\prime})}}\right)-\frac{B^{*}}{\sqrt{N^{\prime}+1}}.
	\end{split}
\end{equation}
Finally,
substituting (\ref{T-U}) and (\ref{jsk-ixiyi-awgn-berry-essen-fw}) into (\ref{c-goal-a*}),
we have
\begin{equation}\label{uB}
    \begin{split}
	 Q\left(\frac{\gamma-R(d)+N^{\prime}C(P^{\prime})}{\sqrt{V_{d}+N^{\prime}\cdot V(P^{\prime})}}\right)-\frac{B^{*}}{\sqrt{N^{\prime}+1}}\leq \varepsilon^{\prime} +\exp\left(- \gamma\right).
\end{split}
\end{equation}
Re-written (\ref{uB}) as
{\small \begin{equation}\label{jsk-ixiyi-awgn-berry-essen--dl}
	\begin{split}
	N^{\prime}\cdot C(P^{\prime})-R(d) \geq \sqrt{V_{d}+N^{\prime}\cdot V(P^{\prime})}\cdot& Q^{-1}\left(\varepsilon^{\prime} +\exp(-\gamma)+\frac{B^{*}}{\sqrt{N^{\prime}+1}}\right)-\gamma,
	\end{split}
\end{equation}}%
which completes the proof.

\subsection{ Proof of step $(c^{*})$ in (\ref{c-goal})}\label{sec4-3-3-3}

Combining (\ref{jsk-ixiyi-awgn-berry-essen--dl})
with $N^{\prime}=N+1$, $P^{\prime} = \frac{P}{1 - \varepsilon - \frac{\zeta}{\sqrt{N}}}$ and $\varepsilon^{\prime} = 1 - \frac{\zeta}{\sqrt{N}}$,
we have
{\small \begin{equation}\label{jsk-ixiyi-awgn-berry-essen--dl-2}
	\begin{split}
	&N^{\prime}\cdot C\left(P^{\prime}\right)-R(d)
	\geq \sqrt{ V_{d}+N^{\prime}\cdot V\left( P^{\prime} \right)}\cdot Q^{-1}\left(\varepsilon^{\prime}  +\exp(-\gamma)+\frac{B^{*}}{\sqrt{N^{\prime}+1}}\right)-\gamma\\	
	&\overset{(d)}{\Rightarrow} N^{\prime}\cdot C\left(P^{\prime}\right)-R(d)
	\geq -\sqrt{ V_{d}+N^{\prime}\cdot V\left(P^{\prime} \right)}\cdot \Phi^{-1}\left(1 - \frac{1}{\sqrt{N^{\prime}}}\right) -\frac{1}{2}\ln(N^{\prime})\\		
	&\overset{(e)}{\Rightarrow} N^{\prime}\cdot C\left(P^{\prime}\right)-R(d)
	\geq -\sqrt{ V_{d}+N^{\prime}\cdot V\left(P^{\prime} \right)}\cdot \sqrt{\ln(N^{\prime})} -\frac{1}{2}\ln(N^{\prime})\\		
	&\overset{(f)}{\Rightarrow} N^{\prime}\cdot C\left(\frac{P}{1 - \varepsilon }\right)+O\left(\sqrt{N^{\prime}}\right)-R(d)
	\geq  -\sqrt{ V_{d}+N^{\prime}\cdot V\left(\frac{P}{1 - \varepsilon }\right)+O\left(\sqrt{N^{\prime}}\right)}\cdot \sqrt{\ln(N^{\prime})}-\frac{1}{2}\ln(N^{\prime})\\
	&\overset{(g)}{\Rightarrow} N^{\prime}\cdot C\left(\frac{P}{1 - \varepsilon }\right)+O\left(\sqrt{N^{\prime}}\right)-R(d)
	 \geq -\sqrt{V_{d}+ N^{\prime}\cdot V\left(\frac{P}{1-\varepsilon}\right)}\cdot \sqrt{\ln(N^{\prime})}
	  -\sqrt{ O(\sqrt{N^{\prime}})\cdot \ln(N^{\prime})}-\frac{1}{2}\ln(N^{\prime})\\
	&\overset{(h)}{\Rightarrow} R(d) \leq
     N^{\prime}\cdot C\left(\frac{P}{1 - \varepsilon }\right)
	+\sqrt{V_{d}+ N^{\prime}\cdot V\left(\frac{P}{1-\varepsilon}\right)}\cdot \sqrt{\ln(N^{\prime})}+O\left(\sqrt{N^{\prime}}\right), \\
\end{split}
\end{equation}}%
where $C(x)\triangleq \frac{1}{2}\ln(1+\frac{x}{\sigma_{\eta}^{2}})$, $V(x)\triangleq \frac{x(x+2\sigma_{\eta}^{2})}{2(x+\sigma_{\eta}^{2})^{2}}$,
\\
$(d)$ follows from choosing
{\small \begin{equation}\label{gz }
	\begin{split}
		\gamma=\frac{1}{2}\ln (N^{\prime}),\quad \zeta=\sqrt{N^{\prime}-1}\left(\frac{2}{\sqrt{N^{\prime}}}+\frac{B^{*}}{\sqrt{N^{\prime}+1}}\right),
\end{split}
\end{equation}}%
which results
\begin{equation}\label{Q-1}
	\begin{split}
	&Q^{-1}\left(1 - \frac{\zeta}{\sqrt{N^{\prime}-1}}  +\exp(-\gamma)+\frac{B^{*}}{\sqrt{N^{\prime}+1}}\right)
	=Q^{-1}\left(1 - \frac{1}{\sqrt{N^{\prime}}}\right)
	=-\Phi^{-1}\left(1 - \frac{1}{\sqrt{N^{\prime}}}\right).
\end{split}
\end{equation}
$(e)$ follows from let
\begin{equation}\label{phi}
	\begin{split}
	\Phi^{-1}\left(1 - \frac{1}{\sqrt{N^{\prime}}}\right)\triangleq t_{N^{\prime}},
\end{split}
\end{equation}
and re-written this as
\begin{equation}\label{TN-0}
    \begin{split}
        \frac{1}{\sqrt{N^{\prime}}} = 1 - \Phi(t_{N^{\prime}}).
    \end{split}
\end{equation}
Applying the Chernoff bound ($1 - \Phi(t) \leq e^{-t^{2}/2}$ for $t > 0$) in (\ref{TN-0}), we obtain for $N^{\prime}\geq 5$
\begin{equation}\label{TN-1}
    \begin{split}
        \frac{1}{\sqrt{N^{\prime}}} = 1 - \Phi(t_{N^{\prime}}) \leq e^{-t_{N^{\prime}}^{2}/2},
    \end{split}
\end{equation}
which implies
\begin{equation}\label{TN-3}
    \begin{split}
t_{N^{\prime}} \leq \sqrt{\ln (N^{\prime})}.
    \end{split}
\end{equation}
For the remaining cases $2 \leq N^{\prime} \leq 4$, it is not difficult to show that the inequality $t_{N^{\prime}} \leq \sqrt{\ln N^{\prime}}$ also holds.
Therefore, substituting (\ref{TN-3}) into (\ref{phi}), we have
\begin{equation}\label{TN-2}
    \begin{split}
\Phi^{-1}\left(1 - \frac{1}{\sqrt{N^{\prime}}}\right)  \leq \sqrt{\ln (N^{\prime})}.
    \end{split}
\end{equation}
$(f)$ follows from that
using the Taylor expansion $(1 - x)^{-1} = 1 + O(x)$ as $x\to 0$, we have
\begin{equation}\label{1-e-N}
    \begin{split}
        P^{\prime}&=\frac{P}{1 - \varepsilon -\frac{\zeta}{\sqrt{N}}}
		=\frac{P}{1 - \varepsilon -\frac{\zeta}{\sqrt{N^{\prime}-1}}}
        = \frac{P}{1 - \varepsilon} \cdot \left(1 - \frac{\zeta}{(1-\varepsilon)\cdot \sqrt{N^{\prime}-1}}\right)^{-1} \\
        &= \frac{P}{1 - \varepsilon} \cdot \left(1 + O\left(\frac{1}{\sqrt{N^{\prime}-1}}\right)\right)
        = \frac{P}{1 - \varepsilon} + O\left(\frac{1}{\sqrt{N^{\prime}-1}}\right)\\
		&= \frac{P}{1 - \varepsilon} + O\left(\frac{1}{\sqrt{N^{\prime}}}\right).
    \end{split}
\end{equation}%
Then, the Taylor expansion of $C(P^{\prime})$ around $\frac{P}{1-\varepsilon}$is given by
{\small \begin{equation}\label{CT}
    \begin{split}
        C\left(P^{\prime}\right)
        &= C\left(\frac{P}{1 - \varepsilon }\right)\cdot \left(P^{\prime}- \frac{P}{1 - \varepsilon }\right)^{0}
		 +C^{\prime}\left(\frac{P}{1 - \varepsilon }\right)\cdot \left(P^{\prime}- \frac{P}{1 - \varepsilon }\right)^{1}
		 +O\left(P^{\prime}- \frac{P}{1 - \varepsilon }\right)\\
		&=C\left(\frac{P}{1 - \varepsilon }\right)+C^{\prime}\left(\frac{P}{1 - \varepsilon }\right)\cdot \frac{1}{\sqrt{N^{\prime}-1}} +O\left(\frac{1}{\sqrt{N^{\prime}-1}}\right)\\
   		&=C\left(\frac{P}{1 - \varepsilon }\right)+ O\left(\frac{1}{\sqrt{N^{\prime}-1}}\right)\\
		&=C\left(\frac{P}{1 - \varepsilon }\right)+ O\left(\frac{1}{\sqrt{N^{\prime}}}\right),
	\end{split}
\end{equation}}%
where $C^{\prime}(x) \triangleq \frac{dC(x)}{dx}$. Similarly, for $V(P^{\prime})$ we have
\begin{equation}\label{VT}
    \begin{split}
        V\left(P^{\prime}\right)
        &= V\left(\frac{P}{1 - \varepsilon }\right) + O\left(\frac{1}{\sqrt{N^{\prime}}}\right);
    \end{split}
\end{equation}
$(g)$ follows from that applying the inequality $\sqrt{x+y} \leq \sqrt{x} + \sqrt{y}$ for $x, y \geq 0$, we obtain
\begin{equation}\label{sqrt-bound}
    \begin{split}
        &\sqrt{V_{d} + N^{\prime}\cdot V\left(\frac{P}{1 - \varepsilon }\right) + O\left(\sqrt{N^{\prime}}\right)} 
		\leq \sqrt{V_{d} + N^{\prime}\cdot V\left(\frac{P}{1 - \varepsilon }\right)} +  \sqrt{O\left(\sqrt{N^{\prime}}\right)};
    \end{split}
\end{equation}
$(h)$ follows from
\begin{equation}\label{sqrt-N}
    \begin{split}
	 O\left(\sqrt{N^{\prime}}\right)
	+\sqrt{ O(\sqrt{N^{\prime}})\cdot \ln(N^{\prime})}+\frac{1}{2}\ln(N^{\prime})
	=O\left(\sqrt{N^{\prime}}\right).
    \end{split}
\end{equation}

\section{Proof of lemma \ref{equalp} in Appendix \ref{App2-1} }\label{App-equalp}
\renewcommand{\theequation}{D\arabic{equation}}

For any $r \in \mathbb{R}$, we have
\begin{equation}\label{eq:decomp}
\begin{split}
			\mathbb{P}\left[X^{*}\leq r\right]&
			=\mathbb{P}\left[(X^{*}\leq r)\bigcap (X^{*}=Y^{*})\right]+\mathbb{P}\left[(X^{*}\leq r)\bigcap (X^{*}\neq Y^{*})\right]\\
			&=\mathbb{P}\left[(Y^{*}\leq r)\bigcap (X^{*}=Y^{*})\right]+\mathbb{P}\left[(X^{*}\leq r)\bigcap (X^{*}\neq Y^{*})\right]\\
			&\overset{(a)}{=}\mathbb{P}\left[(Y^{*}\leq r)\bigcap (X^{*}=Y^{*})\right]+0\\
			&\overset{(b)}{=}\mathbb{P}\left[(Y^{*}\leq r)\bigcap (X^{*}=Y^{*})\right]+\mathbb{P}\left[(Y^{*}\leq r)\bigcap (X^{*}\neq Y^{*})\right]\\
			&=\mathbb{P}\left[(Y^{*}\leq r)\right],\\
\end{split}
\end{equation}
where \\
$(a)$ follows from
\begin{equation}\label{X80}
\begin{split}
	0 \leq \mathbb{P}\left[(X^{*}\leq r)\bigcap (X^{*}\neq Y^{*})\right] \leq\mathbb{P}\left[X^{*}\neq Y^{*}\right] =1-\mathbb{P}\left[X^{*}= Y^{*}\right]=1-1= 0,
\end{split}
\end{equation}
hence $\mathbb{P}\left[(X^{*}\leq r)\bigcap (X^{*}\neq Y^{*})\right]=0$;\\
$(b)$ follows from similar to $(a)$, we have $\mathbb{P}\left[(Y^{*}\leq r)\bigcap (X^{*}\neq Y^{*})\right]=0$.

\section{Proof of lemma \ref{MGF-equality} in Appendix \ref{App2-2} }\label{App-MGFE}
\renewcommand{\theequation}{E\arabic{equation}}
The following shows that the moment generating function (MGF) of $U$ equals that of $T$, i.e.,
\begin{equation}
    \mathbb{E}_{S, X^{N^{\prime}}, Y^{N^{\prime}}} \left[ \exp(tU) \right] = \mathbb{E}_{G,K^{N^{\prime}}}\left[ \exp(tT) \right],
\end{equation}
where $X^{N^{\prime}}\triangleq(X_{1},\ldots,X_{N^{\prime}})$,
$Y^{N^{\prime}}\triangleq(Y_{1},\ldots,Y_{N^{\prime}})$, $K^{N^{\prime}}\triangleq(K_{1},\ldots,K_{N^{\prime}})$,
and $G, K_i (i \in\{1,\ldots,N^{\prime}\})$ are independent normalized Gaussian random variables.
Furthermore, $U$ and $T$ are defined as
\begin{equation}\label{U-fl}
    U \triangleq \frac{S^{2}}{2\sigma_{s}^{2}} - \sum_{i=1}^{N^{\prime}}\Psi(X_{i}, Y_{i}),
\end{equation}
\begin{equation}\label{lambdaxy-fl}
    \Psi(X_{i}, Y_{i}) = \frac{1}{2(P^{\prime}+\sigma_{\eta}^{2})}\left( -\frac{P^{\prime}}{\sigma_{\eta}^{2}}(Y_{i}-X_{i})^{2} + 2X_{i}(Y_{i}-X_{i}) \right),
\end{equation}
and
\begin{equation}\label{T-fl}
    T \triangleq \frac{1}{2} G^{2} + \sum_{i=1}^{N^{\prime}}\left( \frac{P^{\prime}}{2(P^{\prime}+\sigma_{\eta}^{2})}K_{i}^{2} - \frac{\sqrt{P^{\prime}\sigma_{\eta}^{2}}}{P^{\prime}+\sigma_{\eta}^{2}}K_{i} \right).
\end{equation}
Furthermore, under the equal power constraint
\begin{equation}\label{pw-fl}
		\mathbb{P}\left\{\sum\limits_{i=1}^{N^{\prime}} X_{i}^{2}=N^{\prime}P^{\prime}\right\}=1,
\end{equation}
similar to in \cite[equation (11) in P.3]{FBL-feedback-2},
we assume without loss of generality that the joint probability density
\begin{equation}\label{f1}
		f(s,x^{N^{\prime}}, y^{N^{\prime}})=f(s,x^{N^{\prime}}, y^{N^{\prime}}) \cdot \indicator{ \sum\limits_{i=1}^{N^{\prime}} x_{i}^{2}=N^{\prime}P^{\prime}}
\end{equation}
for all $s$, $x^{N^{\prime}}$ and $y^{N^{\prime}}$.

The proof proceeds by computing the MGFs of $U$ and $T$ separately and showing they yield identical expressions.

\subsection{The MGF of $U$}
Following (\ref{U-fl}), the MGF of $U$ is given by
{\small \begin{align}\label{Etlambdaxy-intsxy-fl}
		&\mathbb{E}_{{S, X^{N^{\prime}}, Y^{N^{\prime}}}} \left[ \exp\left(tU\right) \right] \notag \\
		&=\mathbb{E}_{{S, X^{N^{\prime}}, Y^{N^{\prime}}}} \left[ \exp\left(t\left(\frac{S^{2}}{2\sigma_{s}^{2}}-\sum_{i=1}^{N^{\prime}}\Psi(X_{i},Y_{i})\right)\right) \right] \notag \\
		&=\int_{s, x^{N^{\prime}}, y^{N^{\prime}}} f(s,x^{N^{\prime}}, y^{N^{\prime}}) \cdot \exp\left(\frac{ts^{2}}{2\sigma_{s}^{2}}\right)\cdot \exp\left(-t\sum_{i=1}^{N^{\prime}}\Psi(x_{i},y_{i})\right) dy^{N^{\prime}} dx^{N^{\prime}} ds  \notag \\
		&\overset{(a^{\prime})}{=}\int_{s, x^{N^{\prime}}, y^{N^{\prime}}} f(s,x^{N^{\prime}}, y^{N^{\prime}}) \cdot\indicator{ \sum\limits_{i=1}^{N^{\prime}} x_{i}^{2}=N^{\prime}P^{\prime}} \cdot \exp\left(\frac{ts^{2}}{2\sigma_{s}^{2}}\right)
		\cdot \exp\left(-t\sum_{i=1}^{N^{\prime}}\Psi(x_{i},y_{i})\right) dy^{N^{\prime}} dx^{N^{\prime}} ds \notag \\
		&\overset{(b^{\prime})}{=}\int_{s, x^{N^{\prime}}, y^{N^{\prime}}} f(s,x^{N^{\prime}}, y^{N^{\prime}}) \cdot\indicator{ \sum\limits_{i=1}^{N^{\prime}} x_{i}^{2}=N^{\prime}P^{\prime}} \cdot \exp\left(\frac{ts^{2}}{2\sigma_{s}^{2}}\right) \notag \\
		&\qquad \cdot \exp\left(-t\sum_{i=1}^{N^{\prime}}\Psi(x_{i},y_{i})+ \frac{\sigma_{\eta}^{2}t^2}{2(P^{\prime}+\sigma_{\eta}^{2})(P^{\prime}+\sigma_{\eta}^{2}-P^{\prime}t)} \left( N^{\prime}P^{\prime} - \sum\limits_{i=1}^{N^{\prime}} x_{i}^{2} \right)\right) dy^{N^{\prime}} dx^{N^{\prime}} ds \notag \\
		&\overset{(c^{\prime})}{=}\int_{s, x^{N^{\prime}}, y^{N^{\prime}}} f(s,x^{N^{\prime}}, y^{N^{\prime}}) \cdot \exp\left(\frac{ts^{2}}{2\sigma_{s}^{2}}\right)
		\cdot \exp\left(-t\sum_{i=1}^{N^{\prime}}\Psi(x_{i},y_{i})+ \frac{\sigma_{\eta}^{2}t^2}{2(P^{\prime}+\sigma_{\eta}^{2})(P^{\prime}+\sigma_{\eta}^{2}-P^{\prime}t)} \left( N^{\prime}P^{\prime} - \sum\limits_{i=1}^{N^{\prime}} x_{i}^{2} \right)\right) dy^{N^{\prime}} dx^{N^{\prime}} ds \notag \\	
		&\overset{(d^{\prime})}{=}\int_{s}f(s)\cdot \exp\left(\frac{ts^{2}}{2\sigma_{s}^{2}}\right) \int_{x^{N^{\prime}}, y^{N^{\prime}}} f(x^{N^{\prime}}, y^{N^{\prime}} |s)
		 \cdot \exp\left(-t\sum_{i=1}^{N^{\prime}}\Psi(x_{i},y_{i})+ \frac{\sigma_{\eta}^{2}t^2}{2(P^{\prime}+\sigma_{\eta}^{2})(P^{\prime}+\sigma_{\eta}^{2}-P^{\prime}t)} \left( N^{\prime}P^{\prime} - \sum\limits_{i=1}^{N^{\prime}} x_{i}^{2} \right)\right) dy^{N^{\prime}} dx^{N^{\prime}} ds \notag \\
\end{align}}%
where $x^{N^{\prime}}\triangleq(x_{1},\ldots,x_{N^{\prime}})$, $dx^{N^{\prime}}\triangleq dx_{1}\cdot dx_{2} \ldots \cdot dx_{N^{\prime}}$,
$y^{N^{\prime}}\triangleq(y_{1},\ldots,y_{N^{\prime}})$, $dy^{N^{\prime}}\triangleq dy_{1}\cdot dy_{2}\ldots \cdot dy_{N^{\prime}}$,
\\
$(a^{\prime})$ and $(c^{\prime})$ follows from (\ref{pw-fl});\\
$(b^{\prime})$ follows from $\indicator{A=B}=\indicator{A=B}\cdot \exp(B-A)$;\\
$(d^{\prime})$ follows from $f(s,x^{N^{\prime}}, y^{N^{\prime}})=f(s)\cdot f(x^{N^{\prime}}, y^{N^{\prime}}|s)$;


Then, we evaluate the the inner integral over $(x^{N^{\prime}}, y^{N^{\prime}})$ through an iterative reduction \cite{FBL-feedback}.
Consider the following chain of equalities for each $s \in S$ and each $\ell \in \{0,1,\ldots, N^{\prime}-2\}$,
{\small \begin{equation}\label{Etlambdaxy-intxyl}
	\begin{split}
	&\int_{x^{N^{\prime}-\ell}, y^{N^{\prime}-\ell}} f(x^{N^{\prime}-\ell}, y^{N^{\prime}-\ell} |s)
	\cdot \exp\left({-t \left( \sum\limits_{i=1}^{N^{\prime}-\ell} \Psi(x_i, y_i) \right) + \frac{\sigma_{\eta}^{2}t^2}{2(P^{\prime}+\sigma_{\eta}^{2})(P^{\prime}+\sigma_{\eta}^{2}-P^{\prime}t)}  \left( N^{\prime}P^{\prime} - \sum\limits_{i=1}^{N^{\prime}-\ell} x_{i}^{2} \right)}\right) dy^{N^{\prime}-\ell} dx^{N^{\prime}-\ell}\\
	=&\int_{x^{N^{\prime}-\ell-1}, y^{N^{\prime}-\ell-1}} f(x^{N^{\prime}-\ell-1}, y^{N^{\prime}-\ell-1} |s)
	\cdot \exp\left({-t \left( \sum\limits_{i=1}^{N^{\prime}-\ell-1} \Psi(x_i, y_i) \right) + \frac{\sigma_{\eta}^{2}t^2}{2(P^{\prime}+\sigma_{\eta}^{2})(P^{\prime}+\sigma_{\eta}^{2}-P^{\prime}t)}  \left( N^{\prime}P^{\prime} - \sum\limits_{i=1}^{N^{\prime}-\ell-1} x_{i}^{2} \right)} \right)\\
	&\cdot \int_{x_{N^{\prime}-\ell}, y_{N^{\prime}-\ell}} f(x_{N^{\prime}-\ell}, y_{N^{\prime}-\ell}|s,x^{N^{\prime}-\ell-1}, y^{N^{\prime}-\ell-1})
	\cdot \exp\left(-t\Psi(x_{N^{\prime}-\ell}, y_{N^{\prime}-\ell})-\frac{\sigma_{\eta}^{2}t^2}{2(P^{\prime}+\sigma_{\eta}^{2})(P^{\prime}+\sigma_{\eta}^{2}-P^{\prime}t)}  x_{N^{\prime}-\ell}^{2}\right)\\
	&\,\, dy_{N^{\prime}-\ell} dx_{N^{\prime}-\ell}dy^{N^{\prime}-\ell-1} dx^{N^{\prime}-\ell-1}\\
	=&\int_{x^{N^{\prime}-\ell-1}, y^{N^{\prime}-\ell-1}} f(x^{N^{\prime}-\ell-1}, y^{N^{\prime}-\ell-1} |s)
	\cdot \exp\left({-t \left( \sum\limits_{i=1}^{N^{\prime}-\ell-1} \Psi(x_i, y_i) \right) + \frac{\sigma_{\eta}^{2}t^2}{2(P^{\prime}+\sigma_{\eta}^{2})(P^{\prime}+\sigma_{\eta}^{2}-P^{\prime}t)}  \left( N^{\prime}P^{\prime} - \sum\limits_{i=1}^{N^{\prime}-\ell-1} x_{i}^{2} \right)} \right)\\
	&\cdot \int_{x_{N^{\prime}-\ell}} f(x_{N^{\prime}-\ell}|s, x^{N^{\prime}-\ell-1}, y^{N^{\prime}-\ell-1})
	\cdot \exp\left(-\frac{\sigma_{\eta}^{2}t^2}{2(P^{\prime}+\sigma_{\eta}^{2})(P^{\prime}+\sigma_{\eta}^{2}-P^{\prime}t)}  x_{N^{\prime}-\ell}^{2}\right)\\
	&\cdot \int_{y_{N^{\prime}-\ell}}f(y_{N^{\prime}-\ell}|s, x^{N^{\prime}-\ell}, y^{N^{\prime}-\ell-1})\exp\left(-t\Psi(x_{N^{\prime}-\ell}, y_{N^{\prime}-\ell})\right)dy_{N^{\prime}-\ell} dx_{N^{\prime}-\ell}dy^{N^{\prime}-\ell-1} dx^{N^{\prime}-\ell-1}.\\
\end{split}
\end{equation}}%
Then compute the innermost integral $y_{N^{\prime}-\ell}$,
{\small \begin{align}\label{Etlambdaxy-inty}
	&\int_{y_{N^{\prime}-\ell}}f(y_{N^{\prime}-\ell}|s, x^{N^{\prime}-\ell}, y^{N^{\prime}-\ell-1})
	\cdot \exp\left(-t\Psi (x_{N^{\prime}-\ell}, y_{N^{\prime}-\ell})\right)dy_{N^{\prime}-\ell} \notag \\
	&\overset{(a)}{=}\int_{y_{N^{\prime}-\ell}}f(y_{N^{\prime}-\ell}|x_{N^{\prime}-\ell})
	\cdot \exp\left(-t\Psi (x_{N^{\prime}-\ell}, y_{N^{\prime}-\ell})\right)dy_{N^{\prime}-\ell}\notag \\
	&\overset{(b)}{=}\int_{y_{N^{\prime}-\ell}}\frac{1}{\sqrt{2\pi \sigma_{\eta}^{2}}}
	\cdot \exp\left(-\frac{(y_{N^{\prime}-\ell}-x_{N^{\prime}-\ell})^{2}}{2\sigma_{\eta}^{2}}\right)
	\cdot\exp\left(\frac{-t}{2(P^{\prime}+\sigma_{\eta}^{2})}\left(-\frac{P^{\prime}}{\sigma_{\eta}^{2}}(y_{N^{\prime}-\ell}-x_{N^{\prime}-\ell})^{2}+2x_{N^{\prime}-\ell}(y_{N^{\prime}-\ell}-x_{N^{\prime}-\ell})\right)\right)dy_{N^{\prime}-\ell} \notag \\
	&\overset{(c)}{=}\int_{v}\frac{1}{\sqrt{2\pi \sigma_{\eta}^{2}}}
	\cdot\exp\left(-\frac{v^{2}}{2\sigma_{\eta}^{2}}\right)
	\cdot\exp\left(\frac{-t}{2(P^{\prime}+\sigma_{\eta}^{2})}\left(-\frac{P^{\prime}}{\sigma_{\eta}^{2}}v^{2}+2v x_{N^{\prime}-\ell}\right)\right)dv  \notag \\
	&=\int_{v}\frac{1}{\sqrt{2\pi \sigma_{\eta}^{2}}}\exp\left(\left(\frac{P^{\prime}t-P^{\prime}-\sigma_{\eta}^{2}}{2\sigma_{\eta}^{2}(P^{\prime}+\sigma_{\eta}^{2})}\right)v^{2}-\frac{tx_{N^{\prime}-\ell}}{P^{\prime}+\sigma_{\eta}^{2}}v\right)dv \notag \\
	&=\int_{v}\frac{1}{\sqrt{2\pi \sigma_{\eta}^{2}}}\exp\left(\left(\frac{P^{\prime}t-P^{\prime}-\sigma_{\eta}^{2}}{2\sigma_{\eta}^{2}(P^{\prime}+\sigma_{\eta}^{2})}\right)\left(v^{2}-\frac{2t\sigma_{\eta}^{2}x_{N^{\prime}-\ell}}{P^{\prime}t-P^{\prime}-\sigma_{\eta}^{2}}v+\left(\frac{t\sigma_{\eta}^{2}x_{N^{\prime}-\ell}}{P^{\prime}t-P^{\prime}-\sigma_{\eta}^{2}}\right)^{2}-\left(\frac{t\sigma_{\eta}^{2}x_{N^{\prime}-\ell}}{P^{\prime}t-P^{\prime}-\sigma_{\eta}^{2}}\right)^{2}\right)\right)dv  \notag \\
	&=\int_{v}\frac{1}{\sqrt{2\pi \sigma_{\eta}^{2}}}\exp\left(\left(\frac{P^{\prime}t-P^{\prime}-\sigma_{\eta}^{2}}{2\sigma_{\eta}^{2}(P^{\prime}+\sigma_{\eta}^{2})}\right)\left(v-\frac{t\sigma_{\eta}^{2}x_{N^{\prime}-\ell}}{P^{\prime}t-P^{\prime}-\sigma_{\eta}^{2}}\right)^{2}\right)\cdot \exp\left( \frac{t^{2}\sigma_{\eta}^{2}x_{N^{\prime}-\ell}^{2}}{2(P^{\prime}+\sigma_{\eta}^{2})(P^{\prime}+\sigma_{\eta}^{2}-P^{\prime}t)}\right)dv  \notag\\
	&=\frac{1}{\sqrt{2\pi \sigma_{\eta}^{2}}}\cdot \sqrt{\frac{2\pi \sigma_{\eta}^{2}(P^{\prime}+\sigma_{\eta}^{2})}{P^{\prime}+\sigma_{\eta}^{2}-P^{\prime}t}}\cdot \exp\left( \frac{t^{2}\sigma_{\eta}^{2}x_{N^{\prime}-\ell}^{2}}{2(P^{\prime}+\sigma_{\eta}^{2})(P^{\prime}+\sigma_{\eta}^{2}-P^{\prime}t)}\right)  \notag\\
	&\,\,\cdot \int_{v}\sqrt{\frac{P^{\prime}+\sigma_{\eta}^{2}-P^{\prime}t}{2\pi \sigma_{\eta}^{2}(P^{\prime}+\sigma_{\eta}^{2})}}\exp\left(\left(\frac{P^{\prime}t-P^{\prime}-\sigma_{\eta}^{2}}{2\sigma_{\eta}^{2}(P^{\prime}+\sigma_{\eta}^{2})}\right)\left(v-\frac{t\sigma_{\eta}^{2}x_{N^{\prime}-\ell}}{P^{\prime}t-P^{\prime}-\sigma_{\eta}^{2}}\right)^{2}\right)dv  \notag\\
	&= \sqrt{\frac{P^{\prime}+\sigma_{\eta}^{2}}{P^{\prime}+\sigma_{\eta}^{2}-P^{\prime}t}}\cdot \exp\left( \frac{t^{2}\sigma_{\eta}^{2}x_{N^{\prime}-\ell}^{2}}{2(P^{\prime}+\sigma_{\eta}^{2})(P^{\prime}+\sigma_{\eta}^{2}-P^{\prime}t)}\right),
\end{align}}%
where\\
$(a)$ follows from the channel is memoryless;\\
$(b)$ follows from (\ref{lambdaxy-fl});\\
$(c)$ follows fm $v=y_{N^{\prime}-\ell}-x_{N^{\prime}-\ell}$ and $dv=dy_{N^{\prime}-\ell}$;\\

Furthermore, substituting (\ref{Etlambdaxy-inty}) into (\ref{Etlambdaxy-intxyl}) yields
{\small \begin{equation}\label{Etlambdaxy-intxyl-2-fl}
	\begin{split}
	&\int_{x^{N^{\prime}-\ell}, y^{N^{\prime}-\ell}} f(x^{N^{\prime}-\ell}, y^{N^{\prime}-\ell} |s)
	\cdot  \exp\left({-t \left( \sum\limits_{i=1}^{N^{\prime}-\ell} \Psi(x_i, y_i) \right) + \frac{\sigma_{\eta}^{2}t^2}{2(P^{\prime}+\sigma_{\eta}^{2})(P^{\prime}+\sigma_{\eta}^{2}-P^{\prime}t)}  \left( N^{\prime}P^{\prime} - \sum\limits_{i=1}^{N^{\prime}-\ell} x_{i}^{2} \right)}\right) dy^{N^{\prime}-\ell} dx^{N^{\prime}-\ell}\\
	=&\sqrt{\frac{P^{\prime}+\sigma_{\eta}^{2}}{P^{\prime}+\sigma_{\eta}^{2}-P^{\prime}t}}\cdot\int_{x^{N^{\prime}-\ell-1}, y^{N^{\prime}-\ell-1}} f(x^{N^{\prime}-\ell-1}, y^{N^{\prime}-\ell-1} |s)\\
	& \,\cdot \exp\left({-t \left( \sum\limits_{i=1}^{N^{\prime}-\ell-1} \Psi(x_i, y_i) \right) + \frac{\sigma_{\eta}^{2}t^2}{2(P^{\prime}+\sigma_{\eta}^{2})(P^{\prime}+\sigma_{\eta}^{2}-P^{\prime}t)}  \left( N^{\prime}P^{\prime} - \sum\limits_{i=1}^{N^{\prime}-\ell-1} x_{i}^{2} \right)} \right)
	dy^{N^{\prime}-\ell-1} dx^{N^{\prime}-\ell-1},
\end{split}
\end{equation}}%
which implies that applying the recurrence $N^{\prime}-1$ times reduces the $N^{\prime}$-fold integral to a single integral over $(x_{1},y_{1})$,
{\small \begin{align}\label{Etlambdaxy-intxyN}
	&\int_{x^{N^{\prime}}, y^{N^{\prime}}} f(x^{N^{\prime}}, y^{N^{\prime}} |s)
	\cdot \exp\left({-t \sum\limits_{i=1}^{N^{\prime}} \Psi(x_i, y_i) + \frac{\sigma_{\eta}^{2}t^2}{2(P^{\prime}+\sigma_{\eta}^{2})(P^{\prime}+\sigma_{\eta}^{2}-P^{\prime}t)} \left( N^{\prime}P^{\prime} - \sum\limits_{i=1}^{N^{\prime}} x_{i}^{2} \right)}\right) dy^{N^{\prime}} dx^{N^{\prime}}  \notag\\
	&=\left(\frac{P^{\prime}+\sigma_{\eta}^{2}}{P^{\prime}+\sigma_{\eta}^{2}-P^{\prime}t}\right)^{\frac{N^{\prime}-1}{2}}\int_{x_{1}, y_{1}} f(x_{1}, y_{1} |s) \exp\left({-t \left( \Psi (x_1, y_1) \right) + \frac{\sigma_{\eta}^{2}t^2}{2(P^{\prime}+\sigma_{\eta}^{2})(P^{\prime}+\sigma_{\eta}^{2}-P^{\prime}t)}  \left( N^{\prime}P^{\prime} -  x_{1}^{2} \right)}\right) dy_{1} dx_{1}  \notag\\
	&\overset{(d)}{=}\left(\frac{P^{\prime}+\sigma_{\eta}^{2}}{P^{\prime}+\sigma_{\eta}^{2}-P^{\prime}t}\right)^{\frac{N^{\prime}-1}{2}}\int_{x_{1}} f(x_{1}|s) \exp\left({ \frac{\sigma_{\eta}^{2}t^2}{2(P^{\prime}+\sigma_{\eta}^{2})(P^{\prime}+\sigma_{\eta}^{2}-P^{\prime}t)}  \left( N^{\prime}P^{\prime} -  x_{1}^{2} \right)}\right)  \notag\\
	&\qquad \cdot  \int_{y_{1}} f(y_{1}|x_{1}) \exp\left(\frac{-t}{2(P^{\prime}+\sigma_{\eta}^{2})} \left(-\frac{P^{\prime}}{\sigma_{\eta}^{2}}(y_{1}-x_{1})^{2}+2x_{1}(y_{1}-x_{1})\right)\right)dy_{1} dx_{1}  \notag\\
	&\overset{(e)}{=}\left(\frac{P^{\prime}+\sigma_{\eta}^{2}}{P^{\prime}+\sigma_{\eta}^{2}-P^{\prime}t}\right)^{\frac{N^{\prime}-1}{2}}\cdot \left(\frac{P^{\prime}+\sigma_{\eta}^{2}}{P^{\prime}+\sigma_{\eta}^{2}-P^{\prime}t}\right)^{\frac{1}{2}}  \notag\\
	&\qquad \cdot \int_{x_{1}}f(x_{1}|s) \exp\left({ \frac{\sigma_{\eta}^{2}t^2}{2(P^{\prime}+\sigma_{\eta}^{2})(P^{\prime}+\sigma_{\eta}^{2}-P^{\prime}t)}  \left( N^{\prime}P^{\prime} -  x_{1}^{2} \right)}\right) \exp\left(\frac{\sigma_{\eta}^{2}t^2 x_{1}^{2}}{2(P^{\prime}+\sigma_{\eta}^{2})(P^{\prime}+\sigma_{\eta}^{2}-P^{\prime}t)}\right) dx_{1}   \notag\\
	&=\left(\frac{P^{\prime}+\sigma_{\eta}^{2}}{P^{\prime}+\sigma_{\eta}^{2}-P^{\prime}t}\right)^{\frac{N^{\prime}}{2}}\cdot\exp\left(\frac{\sigma_{\eta}^{2}t^2N^{\prime}P^{\prime}}{2(P^{\prime}+\sigma_{\eta}^{2})(P^{\prime}+\sigma_{\eta}^{2}-P^{\prime}t)}  \right)\cdot \int_{x_{1}} f(x_{1}|s) dx_{1}  \notag\\
	&=\left(\frac{P^{\prime}+\sigma_{\eta}^{2}}{P^{\prime}+\sigma_{\eta}^{2}-P^{\prime}t}\right)^{\frac{N^{\prime}}{2}}\cdot\exp\left(\frac{\sigma_{\eta}^{2}t^2N^{\prime}P^{\prime}}{2(P^{\prime}+\sigma_{\eta}^{2})(P^{\prime}+\sigma_{\eta}^{2}-P^{\prime}t)}  \right),
\end{align}}%
where $(d)$ follows from (\ref{lambdaxy-fl}) and $(e)$ is similar to (\ref{Etlambdaxy-inty}).

Now we return to the outer integration over $s$.
Since $S\sim\mathcal{N}(0,\sigma_{s}^{2})$,
{\small \begin{equation}\label{Etlambdaxy-intS}
	\begin{split}
		&\int_{s}f(s)\cdot \exp\left(\frac{ts^{2}}{2\sigma_{s}^{2}}\right)ds
		=\int_{s}\frac{1}{\sqrt{2\pi \sigma_{s}^{2}}}\exp\left(-\frac{s^{2}}{2\sigma_{s}^{2}}\right)\cdot \exp\left(\frac{ts^{2}}{2\sigma_{s}^{2}}\right)ds
		=\frac{1}{\sqrt{2\pi \sigma_{s}^{2}}}\cdot \sqrt{\frac{2\pi \sigma_{s}^{2}}{1-t}}\cdot \int_{s} \sqrt{\frac{1-t}{2\pi \sigma_{s}^{2}}}\exp\left(\frac{t-1}{2\sigma_{s}^{2}}s\right)ds
		=\left(\frac{1}{1-t}\right)^{\frac{1}{2}}.
	\end{split}
\end{equation}}%
Finally, substituting (\ref{Etlambdaxy-intxyN}) and (\ref{Etlambdaxy-intS}) into (\ref{Etlambdaxy-intsxy-fl}),
the MGF of $U$ is re-written by
{\small \begin{equation}\label{Etlambdaxy-intsxy-f-fl}
	\begin{split}
		&\mathbb{E}_{{S, X^{N^{\prime}}, Y^{N^{\prime}}}} \left[ \exp\left(tU\right) \right]
		=\left(\frac{P^{\prime}+\sigma_{\eta}^{2}}{P^{\prime}+\sigma_{\eta}^{2}-P^{\prime}t}\right)^{\frac{N^{\prime}}{2}}\cdot\exp\left(\frac{\sigma_{\eta}^{2}t^2N^{\prime}P^{\prime}}{2(P^{\prime}+\sigma_{\eta}^{2})(P^{\prime}+\sigma_{\eta}^{2}-P^{\prime}t)}  \right) \cdot \left(\frac{1}{1-t}\right)^{\frac{1}{2}}.
	\end{split}
\end{equation}}

\subsection{The MGF of $T$}
Now we compute $\mathbb{E}_{G,K^{N^{\prime}}} [ \exp(tT) ]$.
Since $G, K_i (i \in\{1,\ldots,N^{\prime}\})$ are independent of each other,
{\small \begin{equation}\label{EGKN-fl}
	\begin{split}
		\mathbb{E}_{G,K^{N^{\prime}}} \left[ \exp\left(tT \right) \right]
		&=\mathbb{E}_{G,K^{N^{\prime}}} \left[ \exp\left(t\left(\frac{1}{2} G^{2}-\sum_{i=1}^{N^{\prime}}(-\frac{P^{\prime}}{2(P^{\prime}+\sigma_{\eta}^{2})}K_{i}^{2}+\frac{\sqrt{P^{\prime}\sigma_{\eta}^{2}}}{P^{\prime}+\sigma_{\eta}^{2}}K_{i})\right) \right) \right]\\
		&=\mathbb{E}_{G} \left[ \exp\left(\frac{t}{2} G^{2}\right) \right]
		\cdot \mathbb{E}_{K^{N^{\prime}}} \left[ \exp\left(\sum_{i=1}^{N^{\prime}}(\frac{P^{\prime}t}{2(P^{\prime}+\sigma_{\eta}^{2})}K_{i}^{2}-\frac{t\sqrt{P^{\prime}\sigma_{\eta}^{2}}}{P^{\prime}+\sigma_{\eta}^{2}}K_{i})\right) \right]\\
		&=\mathbb{E}_{G} \left[ \exp\left(\frac{t}{2} G^{2}\right) \right]
		\cdot\mathbb{E}_{K^{N^{\prime}}} \left[ \prod_{i=1}^{N^{\prime}}\exp\left(\frac{P^{\prime}t}{2(P^{\prime}+\sigma_{\eta}^{2})}K_{i}^{2}-\frac{t\sqrt{P^{\prime}\sigma_{\eta}^{2}}}{P^{\prime}+\sigma_{\eta}^{2}}K_{i} \right)\right]\\		
		&=\mathbb{E}_{G} \left[ \exp\left(\frac{t}{2} G^{2}\right) \right]
		\cdot \prod_{i=1}^{N^{\prime}}\mathbb{E}_{K_{i}} \left[ \exp\left(\frac{P^{\prime}t}{2(P^{\prime}+\sigma_{\eta}^{2})}K_{i}^{2}-\frac{t\sqrt{P^{\prime}\sigma_{\eta}^{2}}}{P^{\prime}+\sigma_{\eta}^{2}}K_{i} \right)\right].\\		
	\end{split}
\end{equation}}%
For $G \sim \mathcal{N}(0,1)$,
{\small \begin{equation}\label{EG-fl}
	\begin{split}
		&\mathbb{E}_{G} \left[ \exp\left(\frac{t}{2} G^{2}\right) \right]
		=\frac{1}{\sqrt{2\pi}} \int_{g}\exp\left(-\frac{g^{2}}{2}\right)\cdot \exp\left(\frac{t}{2} g^{2}\right) dg
	    =\frac{1}{\sqrt{2\pi}} \cdot \sqrt{\frac{2\pi}{1-t}}\int_{g}\sqrt{\frac{1-t}{2\pi}}\cdot\exp\left(\frac{t-1}{2}g^{2}\right)dg
		=\left(\frac{1}{1-t}\right)^{\frac{1}{2}}.
	\end{split}
\end{equation}}%
For  $K_i \sim \mathcal{N}(0,1)$,
{\small \begin{equation}\label{EKi}
	\begin{split}
	&\mathbb{E}_{K_{i}} \left[ \exp\left(\frac{P^{\prime}t}{2(P^{\prime}+\sigma_{\eta}^{2})}K_{i}^{2}-\frac{t\sqrt{P^{\prime}\sigma_{\eta}^{2}}}{P^{\prime}+\sigma_{\eta}^{2}}K_{i} \right)\right]\\
	&=\frac{1}{\sqrt{2\pi}} \int_{k_{i}}\exp\left(-\frac{k_{i}^{2}}{2}\right)\cdot \exp\left(\frac{P^{\prime}t}{2(P^{\prime}+\sigma_{\eta}^{2})}k_{i}^{2}-\frac{t\sqrt{P^{\prime}\sigma_{\eta}^{2}}}{P^{\prime}+\sigma_{\eta}^{2}}k_{i} \right) dk_{i}\\
	&=\frac{1}{\sqrt{2\pi}} \int_{k_{i}}\exp\left(\frac{P^{\prime}t-P^{\prime}-\sigma_{\eta}^{2}}{2(P^{\prime}+\sigma_{\eta}^{2})}k_{i}^{2}-\frac{t\sqrt{P^{\prime}\sigma_{\eta}^{2}}}{P^{\prime}+\sigma_{\eta}^{2}}k_{i} \right) dk_{i}\\
	&=\frac{1}{\sqrt{2\pi}} \int_{k_{i}}\exp\left(\left(\frac{P^{\prime}t-P^{\prime}-\sigma_{\eta}^{2}}{2(P^{\prime}+\sigma_{\eta}^{2})}\right)\left(k_{i}^{2}-\frac{2t\sqrt{P^{\prime}\sigma_{\eta}^{2}}}{P^{\prime}t-P^{\prime}-\sigma_{\eta}^{2}}k_{i} +\left(\frac{t\sqrt{P^{\prime}\sigma_{\eta}^{2}}}{P^{\prime}t-P^{\prime}-\sigma_{\eta}^{2}}\right)^{2}-\left(\frac{t\sqrt{P^{\prime}\sigma_{\eta}^{2}}}{P^{\prime}t-P^{\prime}-\sigma_{\eta}^{2}}\right)^{2} \right)\right) dk_{i}\\
	&=\frac{1}{\sqrt{2\pi}} \exp\left(\frac{t^{2}P^{\prime}\sigma_{\eta}^{2}}{2(P^{\prime}+\sigma_{\eta}^{2})(P^{\prime}+\sigma_{\eta}^{2}-P^{\prime}t)}\right)\cdot\int_{k_{i}}\exp\left(\left(\frac{P^{\prime}t-P^{\prime}-\sigma_{\eta}^{2}}{2(P^{\prime}+\sigma_{\eta}^{2})}\right)\left(k_{i}-\frac{t\sqrt{P^{\prime}\sigma_{\eta}^{2}}}{P^{\prime}t-P^{\prime}-\sigma_{\eta}^{2}}\right)^{2}\right) dk_{i}\\
	&=\frac{1}{\sqrt{2\pi}} \exp\left(\frac{t^{2}P^{\prime}\sigma_{\eta}^{2}}{2(P^{\prime}+\sigma_{\eta}^{2})(P^{\prime}+\sigma_{\eta}^{2}-P^{\prime}t)}\right)\cdot \sqrt{\frac{2\pi (P^{\prime}+\sigma_{\eta}^{2})}{P^{\prime}+\sigma_{\eta}^{2}-P^{\prime}t}}\\
	&\,\,\cdot \int_{k_{i}}\sqrt{\frac{P^{\prime}+\sigma_{\eta}^{2}-P^{\prime}t}{2\pi(P^{\prime}+\sigma_{\eta}^{2})}}\cdot\exp\left(\left(\frac{P^{\prime}t-P^{\prime}-\sigma_{\eta}^{2}}{2(P^{\prime}+\sigma_{\eta}^{2})}\right)\left(k_{i}-\frac{t\sqrt{P^{\prime}\sigma_{\eta}^{2}}}{P^{\prime}t-P^{\prime}-\sigma_{\eta}^{2}}\right)^{2}\right) dk_{i}\\
	&=\sqrt{\frac{ (P^{\prime}+\sigma_{\eta}^{2})}{P^{\prime}+\sigma_{\eta}^{2}-P^{\prime}t}}\cdot \exp\left(\frac{t^{2}P^{\prime}\sigma_{\eta}^{2}}{2(P^{\prime}+\sigma_{\eta}^{2})(P^{\prime}+\sigma_{\eta}^{2}-P^{\prime}t)}\right).\\
\end{split}
\end{equation}}%
Substituting (\ref{EG-fl}) and (\ref{EKi}) into (\ref{EGKN-fl}), the $\mathbb{E}_{G,K^{N^{\prime}}} [ \exp(tT) ]$ is re-written by
\begin{equation}\label{MGF-T-final}
\begin{split}
    \mathbb{E}_{G,K^{N^{\prime}}} [ \exp(tT) ]=
    \left(\frac{P^{\prime}+\sigma_{\eta}^{2}}{P^{\prime}+\sigma_{\eta}^{2}-P^{\prime}t}\right)^{\frac{N^{\prime}}{2}}\cdot\exp\left(\frac{\sigma_{\eta}^{2}t^2N^{\prime}P^{\prime}}{2(P^{\prime}+\sigma_{\eta}^{2})(P^{\prime}+\sigma_{\eta}^{2}-P^{\prime}t)}  \right) \cdot \left(\frac{1}{1-t}\right)^{\frac{1}{2}}.
\end{split}
\end{equation}
Comparing (\ref{Etlambdaxy-intsxy-f-fl}) and (\ref{MGF-T-final}), we observe that
\begin{equation}\label{MGF-T-U-fl}
\begin{split}
\mathbb{E}_{{S, X^{N^{\prime}}, Y^{N^{\prime}}}}[ \exp(tU) ]
    =\mathbb{E}_{G,K^{N^{\prime}}}[ \exp(tT) ]
    =\left(\frac{P^{\prime}+\sigma_{\eta}^{2}}{P^{\prime}+\sigma_{\eta}^{2}-P^{\prime}t}\right)^{\!\frac{N^{\prime}}{2}}
      \exp\left(\frac{\sigma_{\eta}^{2}t^{2}N^{\prime}P^{\prime}}{2(P^{\prime}+\sigma_{\eta}^{2})(P^{\prime}+\sigma_{\eta}^{2}-P^{\prime}t)}\right)
      \left(\frac{1}{1-t}\right)^{\!\frac{1}{2}} .
\end{split}
\end{equation}
Thus the moment generating functions of $U$ and $T$ are identical, which completes the proof.

\section{The proof of lemma \ref{E-Var-jsj-T} in Appendix \ref{App2-2} }\label{App-EVJ}
\renewcommand{\theequation}{F\arabic{equation}}
Note that
{\begin{equation}
    \begin{split}
    &T^{*}_{0} \triangleq \frac{1}{2}(G^2-1),  \quad
    T^{*}_{i} \triangleq A\cdot(K_i^2-1) + B\cdot K_i, \quad 1\leq i \leq N^{\prime},
    \end{split}
\end{equation}}%
where $G, K_i (i \in\{1,\ldots,N^{\prime}\})$ are independent normalized Gaussian random variables
and
\begin{equation}\label{AB}
    A \triangleq \frac{P^{\prime}}{2(P^{\prime}+\sigma_{\eta}^{2})},\qquad
    B \triangleq -\frac{\sqrt{P^{\prime}\sigma_{\eta}^{2}}}{P^{\prime}+\sigma_{\eta}^{2}}.
\end{equation}
For arbitrary normalized Gaussian variable, i.e., $X\sim\mathcal{N}(0,1)$, it is not difficult to check that
{\begin{equation}\label{gsj}
    \begin{split}
      \mathbb{E}[X] = 0, \quad \mathbb{E}[X^2] = 1, \quad \mathbb{E}[X^3] = 0, \quad \mathbb{E}[X^4] = 3,
	  \quad \mathbb{E}[X^6]=15,\quad \mathbb{E}[|X|^3] = 2\sqrt{\frac{2}{\pi}}.
 \end{split}
\end{equation}}%

\subsection{Sum of the Means}
First, for $T^{*}_0$ we obtain
\begin{equation}\label{E-T0-fl}
    \begin{split}
       \mathbb{E}[T^{*}_0] = \mathbb{E}\left[ \frac{1}{2}(G^2-1)\right]
        =\frac{1}{2}\left(\mathbb{E}\left[G^2\right]-1\right)
        =0.
    \end{split}
\end{equation}%

For $j\in\{1,\dots,N^{\prime}\}$, we have
\begin{equation}\label{E-T1N-fl}
    \begin{split}
    \mathbb{E}\left[T^{*}_j\right] &= \mathbb{E}\left[A\cdot(K_j^2 - 1) + B\cdot K_j\right]
    = A\cdot(\mathbb{E}\left[K_j^2 \right]- 1) + B\cdot \mathbb{E}\left[ K_j\right]
    =0
\end{split}
\end{equation}%
Therefore, the Sum of the means is given by
\begin{equation}\label{E-T0-T1N-fl}
    \begin{split}
    \sum_{i=0}^{N^{\prime}}\mathbb{E}[ T^{*}_{i}]=0.
\end{split}
\end{equation}%

\subsection{Sum of the Variances}
First, for $T^{*}_0$ we obtain
\begin{equation}\label{E-variances-fl}
    \begin{split}
        {\rm Var}[T^{*}_0] = \mathbb{E}\left[ \left(\frac{1}{2}(G^2-1)\right)^{2}\right]
        = \mathbb{E}\left[ \frac{1}{4}(G^4-2G^{2}+1)\right]
        = \frac{1}{4}\left(\mathbb{E}\left[G^4\right]-2\mathbb{E}\left[G^{2}\right]+1\right)
        =\frac{1}{2}.
    \end{split}
\end{equation}%
For $j\in\{1,\dots,N^{\prime}\}$, we have
\begin{equation}\label{E-variances-T1Nfl}
    \begin{split}
    {\rm Var}[T^{*}_j] &= \mathbb{E}\left[(T^{*}_j)^2\right]=\mathbb{E}\left[(A\cdot(K_j^2 - 1) + B\cdot  K_j)^2\right]\\
    &=\mathbb{E}\left[ A^2 (K_j^2 - 1)^2 + 2AB K_j (K_j^2 - 1) + B^2 K_j^2\right]\\
    &=\mathbb{E}\left[A^2 (K_j^4 - 2K_j^2 + 1) + 2AB (K_j^3 - K_j) + B^2 K_j^2\right]\\
    &=A^2 (\mathbb{E}[K_j^4] - 2\mathbb{E}[K_j^2] + 1) + 2AB (\mathbb{E}[K_j^3] - \mathbb{E}[K_j]) + B^2 \mathbb{E}[K_j^2]\\
    &= 2A^2 + B^2.
\end{split}
\end{equation}%
Subsequently, substituting (\ref{AB}) into (\ref{E-variances-T1Nfl}), the ${\rm Var}[T^{*}_j]$ is re-written by
\begin{equation}\label{E-variances-T1N-f-fl}
    \begin{split}
    {\rm Var}[T^{*}_j] &= 2A^2 + B^2
    =2\left(\frac{P^{\prime}}{2(P^{\prime}+\sigma_{\eta}^{2})}\right)^{2}+\left(-\frac{\sqrt{P^{\prime}\sigma_{\eta}^{2}}}{P^{\prime}+\sigma_{\eta}^{2}}\right)^{2}
    =\frac{{P^{\prime}}^{2}}{2(P^{\prime}+\sigma_{\eta}^{2})^{2}}+\frac{P^{\prime}\sigma_{\eta}^{2}}{(P^{\prime}+\sigma_{\eta}^{2})^{2}}
    =\frac{P^{\prime}(P^{\prime}+2\sigma_{\eta}^{2})}{2(P^{\prime}+\sigma_{\eta}^{2})^{2}}.
\end{split}
\end{equation}%
Finally, the Sum of the variances is given by
\begin{equation}\label{E-variances-T0N-f-fl}
    \begin{split}
     \sum_{i=0}^{N^{\prime}} {\rm Var}\left[  T^{*}_{i} \right]
     ={\rm Var}[T^{*}_0]+ \sum_{j=1}^{N^{\prime}}{\rm Var}[T^{*}_j]
     =\frac{1}{2}+\frac{N^{\prime}P^{\prime}(P^{\prime}+2\sigma_{\eta}^{2})}{2(P^{\prime}+\sigma_{\eta}^{2})^{2}}
     = V_{d}+N^{\prime}\cdot V(P^{\prime}).
\end{split}
\end{equation}%
where $V_{d}\triangleq \frac{1}{2}$ and $V(P^{\prime})\triangleq \frac{P^{\prime}(P^{\prime}+2\sigma_{\eta}^{2})}{2(P^{\prime}+\sigma_{\eta}^{2})^{2}}$.
\subsection{Sum of the Third Absolute Moments}

For $i \in\{0,\ldots,N^{\prime}\}$, let $\rho_i \triangleq \mathbb{E}[|(T^{*}_i)^3|]$ denotes the third absolute moment, which is strictly positive ($\rho_i \geq 0$).

First, for $T^{*}_0$ we obtain
\begin{equation}\label{jsj-T0-fl}
    \begin{split}
    \rho_0 &=\mathbb{E}[|(T^{*}_0)^{3}|]=\mathbb{E}[|T^{*}_0|^{3}]\\
    &=\mathbb{E}\left[\left\lvert \frac{1}{2}(G^2-1)\right\rvert^{3} \right]
    = \frac{1}{8}\mathbb{E}[|G^2-1|^3] \\
    &\overset{(a)}{\leq } \frac{1}{8}\mathbb{E}[(G^2+1)^3] \\
    &= \frac{1}{8}\mathbb{E}[G^6 + 3G^4 + 3G^2 + 1] \\
    &= \frac{1}{8}\left(\mathbb{E}[G^6] + 3\mathbb{E}[G^4] + 3\mathbb{E}[G^2] + 1\right)\\
    &\overset{(b)}{= }\frac{1}{8}(15 + 3\cdot 3 + 3\cdot 1 + 1) \\
    &=  3.5 < +\infty,
\end{split}
\end{equation}
where $(a)$ follows from triangle inequality $|x+y|\leq |x|+|y|$.
and $(b)$ follows from that for $G \sim \mathcal{N}(0,1)$, we have $\mathbb{E}[G^2]=1$, $\mathbb{E}[G^4]=3$, and $\mathbb{E}[G^6]=15$.

Note that $\rho_i=\rho_1$ for all $1 \leq i\leq N^{\prime}$.
Then, for $T^{*}_1$ we obtain
\begin{equation}\label{jsj-T1-fl}
    \begin{split}
       \rho_1 &=\mathbb{E}[|(T^{*}_1)^{3}|]=\mathbb{E}[|T^{*}_1|^{3}]\\
       &= \mathbb{E}\left[\left\lvert A\cdot(K_1^{2}-1) + B\cdot K_1\right\rvert^{3} \right] \\
       & \overset{(c)}{\leq} \mathbb{E}\left[\left(\left\lvert A\cdot (K_1^{2}-1) \right\rvert+ \left\lvert B\cdot K_1\right\rvert\right)^{3} \right]\\
       &\overset{(d)}{\leq}\mathbb{E}\left[8\left\lvert A\cdot (K_1^{2}-1) \right\rvert^{3}+ 8\left\lvert B\cdot K_1\right\rvert^{3} \right]\\
       &=8|A|^{3}\mathbb{E}\left[\left\lvert K_1^{2}-1 \right\rvert^{3}\right]+8|B|^{3}\mathbb{E}\left[\left\lvert K_1 \right\rvert^{3}\right]\\
       & \overset{(e)}{\leq} 8|A|^{3}\mathbb{E}\left[(K_1^2+1)^3 \right]+8|B|^{3}\mathbb{E}\left[\left\lvert K_1 \right\rvert^{3}\right]\\
       & = 8|A|^{3}\mathbb{E}\left[K_1^6 + 3K_1^4 + 3K_1^2 + 1 \right]+8|B|^{3}\mathbb{E}\left[\left\lvert K_1 \right\rvert^{3}\right]\\
       & = 8|A|^{3}\left(\mathbb{E}[K_1^6] + 3\mathbb{E}[K_1^4] + 3\mathbb{E}[K_1^2] + 1\right)+8|B|^{3}\mathbb{E}\left[\left\lvert K_1 \right\rvert^{3}\right]\\
       & \overset{(f)}{=} 8|A|^3 \cdot 28 + 8|B|^3 \cdot 2\sqrt{\frac{2}{\pi}}\\
       &< +\infty.
    \end{split}
\end{equation}%
where $(c)$ and $(e)$ follows from triangle inequality $|x+y|\leq |x|+|y|$,
$(d)$ follows from $(x+y)^3 \leq 8(x^3+y^3)$ for nonnegative $x,y$,
and $(f)$ follows from that for $K_{1} \sim \mathcal{N}(0,1)$, we have $\mathbb{E}[K_{1}^2]=1$, $\mathbb{E}[K_{1}^4]=3$, $\mathbb{E}[K_{1}^6]=15$ and $\mathbb{E}[|K_i|^3] = 2\sqrt{2/\pi}$.

Hence the total third absolute moment of the sum is
\begin{equation}\label{total-rho}
   0 \leq \rho \triangleq \sum_{0}^{N^{\prime}}\rho_{i}= \rho_0+ N^{\prime}\cdot \rho_{1} < +\infty.
\end{equation}

\section{The Proof of Lemma \ref{pu} in Section \ref{sec4-2}}\label{App-pu}
\setcounter{equation}{0}
\renewcommand{\theequation}{G\arabic{equation}}
From our scheme stated in Section \ref{sec4-2}, it is not difficult to show that
{\small \begin{equation}\label{1to1-s-d-fl}
	\begin{split}
	\mathcal{P}_{d, N} &\overset{(a)}{\leq} 2Q\left(\sqrt{\frac{d}{\alpha_{N}}}\right)
	\overset{(b)}{=} 2Q\left( \left(\frac{d}{\sigma_{s}^{2}}\right)^{\frac{1}{2}}\cdot \left( 1+\frac{P}{\sigma_{\eta}^{2}}\right)^{\frac{N}{2}}\right)\\
	&=2Q\left(\exp\left(-\frac{1}{2}\ln\left(\frac{\sigma_{s}^{2}}{d}\right)+\frac{N}{2}\ln\left(1+\frac{P}{\sigma_{\eta}^{2}}\right)\right)\right)\\
	&\overset{(c)}{=}2Q\left(\exp\left(-R(d)+N\cdot C(P)\right)\right),\\
\end{split}
\end{equation}}%
where
$(a)$ follows from
{\small \begin{equation}\label{1to1-pr}
	\begin{split}
	\mathcal{P}_{d, N}&=\mathbb{P}\left\{(S-\widehat{S}_{N})^{2} \geq d \right\} =\mathbb{P}\left\{ \epsilon_{N}^{2} \geq d \right\}
	=\mathbb{P}\left\{ (\epsilon_{N}\geq \sqrt{d} )\bigcup  (\epsilon_{N}\leq -\sqrt{d} )\right\} \\
	&\leq \mathbb{P}\left\{ \epsilon_{N}\geq \sqrt{d} \right\}+\mathbb{P}\left\{ \epsilon_{N}\leq -\sqrt{d} \right\}
	=2Q\left(\sqrt{\frac{d}{\alpha_{N}}}\right);
\end{split}
\end{equation}}%
$(b)$ follows from for $i\in \{1,\ldots,N\}$,
{\small \begin{equation}\label{1to1-ai}
	\alpha_{i}\triangleq  {\rm Var}(\epsilon_{i}) = \sigma_{s}^{2}\cdot\left(\frac{\sigma_{\eta}^{2}}{P+\sigma_{\eta}^{2}}\right)^{i};\\
\end{equation}}%
$(c)$ follows from {\small $C(P)\triangleq \frac{1}{2}\ln\left(1+\frac{P}{\sigma_{\eta}^{2}}\right)$} and {\small $R(d) \triangleq \frac{1}{2}\ln\left(\frac{\sigma_{s}^{2}}{d}\right)$}.\\

\section{The Proof of Lemma \ref{Ln-l} in Section \ref{sec4-2}}\label{App-ln}
\setcounter{equation}{0}
\renewcommand{\theequation}{H\arabic{equation}}

For the eavesdropper, the information leakage rate is given by
{\small \begin{equation}\label{1to1-Ln-fl}
	\begin{aligned}
		L_{N}&=\frac{I(S;Z^{N},\widetilde{Z}^{N})}{N}=\frac{h(S)-h(S|Z^{N},\widetilde{Z}^{N})}{N}.\\
	\end{aligned}
\end{equation}}%
Then, we obtain the upper bound on $L_{N}$ by establishing the lower bound of $h(S|Z^{N},\widetilde{Z}^{N})$.
{\small
\begin{align}\label{1to1-h(s|Z)}
		h(S|&Z^{N},\widetilde{Z}^{N})\notag \\
		&=h(S|X_{1}+\eta_{e,1},\ldots,X_{N}+\eta_{e,N},X_{1}+\eta_{1}+\widetilde{\eta}_{e,1},\ldots,X_{N}+\eta_{N}\notag
		 +\widetilde{\eta}_{e,N})\notag \\
		&\overset{(a)}{\geq} h(S| X_{1}+\eta_{e,1},\ldots,X_{N}+\eta_{e,N},X_{1}+\eta_{1}+\widetilde{\eta}_{e,1}, \ldots,X_{N}+\eta_{N}
		 +\widetilde{\eta}_{e,N},\eta_{1}, \ldots, \eta_{N})\notag \\
		&\overset{(b)}{=}h(S|\eta_{e,1}+\lambda S,\ldots,\eta_{e,N}-\lambda\cdot \kappa^{\frac{N-1}{2}}\cdot S+\left(1-\kappa \right)\cdot\kappa^{\frac{N-3}{2}}\cdot\eta_{1}
		 -\sum_{k=2}^{N-1}\left(1-\kappa \right)\cdot\kappa^{\frac{N-k-2}{2}}\cdot\eta_{k},\eta_{1}+\widetilde{\eta}_{e,1}+\lambda S,\ldots,\eta_{N}+\widetilde{\eta}_{e,N}\notag \\
		&\quad - \lambda \cdot\kappa^{\frac{N-1}{2}}\cdot S+(1-\kappa)\cdot\kappa^{\frac{N-3}{2}}\cdot\eta_{1}-\sum_{k=2}^{N-1} \left(1-\kappa \right)\cdot\kappa^{\frac{N-k-2}{2}}\cdot\eta_{k},
		 \eta_{1}, \ldots, \eta_{N})\notag \\
		&=h(S|\eta_{1}, \ldots, \eta_{N},\eta_{e,1}+\lambda S,\ldots, \eta_{e,N}-\lambda \cdot\kappa^{\frac{N-1}{2}}\cdot S,\widetilde{\eta}_{e,1}+\lambda S,
		 \ldots,\widetilde{\eta}_{e,N}-\lambda \cdot\kappa^{\frac{N-1}{2}}\cdot S)\notag \\
		&\overset{(c)}{=}h(S|\eta_{e,1}+\lambda S,\ldots,\eta_{e,N}-\lambda \cdot\kappa^{\frac{N-1}{2}}\cdot S,\widetilde{\eta}_{e,1}+\lambda S,\ldots, \widetilde{\eta}_{e,N}
		  -\lambda \cdot\kappa^{\frac{N-1}{2}}\cdot S)\notag \\
	    &=h(S)-I(S;\eta_{e,1}+\lambda S,\ldots,\eta_{e,N}-\lambda \cdot\kappa^{\frac{N-1}{2}}\cdot S, \widetilde{\eta}_{e,1}+\lambda S,\ldots,
		 \widetilde{\eta}_{e,N}-\lambda \cdot\kappa^{\frac{N-1}{2}}\cdot S)\notag \\
		&=h(S)-h(\eta_{e,1}+\lambda S,\ldots,\eta_{e,N}-\lambda \cdot\kappa^{\frac{N-1}{2}}\cdot S,\widetilde{\eta}_{e,1}+\lambda S,\ldots,
		 \widetilde{\eta}_{e,N}-\lambda \cdot\kappa^{\frac{N-1}{2}}\cdot S) \notag \\
		&\quad +h(\eta_{e,1}+\lambda S,\ldots,\eta_{e,N}-\lambda \cdot\kappa^{\frac{N-1}{2}}\cdot S,
		\widetilde{\eta}_{e,1}+\lambda S,\ldots,\widetilde{\eta}_{e,N}-\lambda \cdot\kappa^{\frac{N-1}{2}}\cdot S|S)\notag \\
		&=h(S)-h(\eta_{e,1}+\lambda S,\ldots,\eta_{e,N}-\lambda \cdot\kappa^{\frac{N-1}{2}}\cdot S, \widetilde{\eta}_{e,1}+\lambda S,\ldots,
		 \widetilde{\eta}_{e,N}-\lambda \cdot\kappa^{\frac{N-1}{2}}\cdot S)+h(\eta_{e,1},\ldots,\eta_{e,N},\widetilde{\eta}_{e,1},\ldots,\widetilde{\eta}_{e,N}|S)\notag \\
		&\overset{(d)}{=}h(S)-h(\eta_{e,1}+\lambda S,\ldots,\eta_{e,N}-\lambda \cdot\kappa^{\frac{N-1}{2}}\cdot S,\widetilde{\eta}_{e,1}+\lambda S,\ldots,
		 \widetilde{\eta}_{e,N}-\lambda \cdot\kappa^{\frac{N-1}{2}}\cdot S) +h(\eta_{e,1},\ldots,\eta_{e,N},\widetilde{\eta}_{e,1},\ldots,\widetilde{\eta}_{e,N})\notag \\
		&\overset{(e)}{=}h(S)+\sum_{i=1}^{N}h(\eta_{e,i})+\sum_{i=1}^{N}h(\widetilde{\eta}_{e,i})-h(\eta_{e,1}+\lambda S,\ldots, \eta_{e,N}
		 -\lambda \cdot\kappa^{\frac{N-1}{2}}\cdot S, \widetilde{\eta}_{e,1}+\lambda S,\ldots,\widetilde{\eta}_{e,N}-\lambda \cdot\kappa^{\frac{N-1}{2}}\cdot S)\notag \\
		&\overset{(f)}{\geq}h(S)+\sum_{i=1}^{N}h(\eta_{e,i})+\sum_{i=1}^{N}h(\widetilde{\eta}_{e,i})-h(\eta_{e,1}+\lambda S)-h(\widetilde{\eta}_{e,1}+\lambda S)
		 -\sum_{i=2}^{N}h(-\lambda \cdot\kappa^{\frac{i-1}{2}}\cdot S+\eta_{e,i})-\sum_{i=2}^{N}h(-\lambda \cdot\kappa^{\frac{i-1}{2}}\cdot S+\widetilde{\eta}_{e,i})\notag \\
		&\overset{(g)}{\geq}h(S)+\frac{N}{2}\ln(2\pi e \sigma_{e}^{2})-\frac{1}{2}\sum_{i=1}^{N}\ln(2 \pi e(P\cdot\kappa^{i-1}+\sigma_{e}^{2}))
		 +\frac{N}{2}\ln(2\pi e \widetilde{\sigma}_{e}^{2})-\frac{1}{2}\sum_{i=1}^{N}\ln(2 \pi e(P\cdot\kappa^{i-1}+\widetilde{\sigma}_{e}^{2}))\notag \\
		&=h(S)-\frac{1}{2}\sum_{i=1}^{N}\left(\ln(\frac{P}{\sigma_{e}^{2}}\cdot\kappa^{i-1}+1)+\ln(\frac{P}{\widetilde{\sigma}_{e}^{2}}\cdot\kappa^{i-1}+1)\right)\notag \\
		&\overset{(h)}{\geq}h(S)-\frac{1}{2}\sum_{i=1}^{N}\frac{P}{\sigma_{e}^{2}}\cdot\kappa^{i-1}-\frac{1}{2}\sum_{i=1}^{N}\frac{P}{\widetilde{\sigma}_{e}^{2}}\cdot\kappa^{i-1}\notag \\
		&=h(S)-\frac{1}{2}\left(\frac{P(\sigma_{e}^{2}+\widetilde{\sigma}_{e}^{2})}{\sigma_{e}^{2}\widetilde{\sigma}_{e}^{2}}\right)\cdot\sum_{i=1}^{N}\kappa^{i-1}\notag \\
        &\overset{(i)}{=}h(S)-\frac{1}{2}\left(\frac{P(\sigma_{e}^{2}+\widetilde{\sigma}_{e}^{2})}{\sigma_{e}^{2}\widetilde{\sigma}_{e}^{2}}\right)\cdot\frac{P+\sigma_{\eta}^{2}}{P}\cdot\left(1-\kappa^{N}\right)\notag \\
		&=h(S)-\frac{1}{2}\left(\frac{(P+\sigma_{\eta}^{2})(\sigma_{e}^{2}+\widetilde{\sigma}_{e}^{2})}{\sigma_{e}^{2}\widetilde{\sigma}_{e}^{2}}\right)\left(1-\kappa^{N}\right),
\end{align}}%
where $\lambda \triangleq \sqrt{\frac{P}{\sigma_{s}^{2}}}$, $\kappa \triangleq \frac{\sigma_{\eta}^{2}}{P+\sigma_{\eta}^{2}}$, $\tau \triangleq \frac{\sqrt{P\sigma_{s}^{2}}}{P+\sigma_{\eta}^{2}}$,\\
$(a)$ follows from conditions reduce entropy;\\
$(b)$ follows from the fact that, given the coding scheme in Section \ref{sec4-2}, it is not difficult to check that
{\small \begin{equation} \label{1to1-Xn}
	X_{n}=
	\left\{
		\begin{array}{lr}
			\lambda S,\,\, n=1, \, \vspace{8pt}&\\
			\sqrt{\frac{P}{\alpha_{n-1}}}\epsilon_{n-1}, \, \, n\in\{2,\ldots,N\},&
		\end{array}
	\right.
\end{equation}}%
where
{\small \begin{equation} \label{1to1-An}
	\begin{aligned}
	\sqrt{\frac{P}{\alpha_{n-1}}}\epsilon_{n-1}=&-\lambda \cdot\kappa^{\frac{n-1}{2}}\cdot S +\left(1-\kappa \right)\cdot\kappa^{\frac{n-3}{2}}\cdot\eta_{1}
	-\sum_{k=2}^{n-1}\left(1-\kappa \right)\cdot\kappa^{\frac{n-k-2}{2}}\cdot\eta_{k};
	\end{aligned}
\end{equation}}
$(c)$ follows from the fact that $\eta_{1}, \ldots, \eta_{N}$ are independent of $S,\eta_{e,1}+\lambda S,\ldots, \eta_{e,N}-\lambda \cdot\kappa^{\frac{N-1}{2}}\cdot S,\widetilde{\eta}_{e,1}+\lambda S,\ldots, \widetilde{\eta}_{e,N}-\lambda \cdot\kappa^{\frac{N-1}{2}}\cdot S$;\\
$(d)$ follows from the fact that $S$ is independent of $\eta_{e,1},\ldots,\eta_{e,N},\widetilde{\eta}_{e,1},\ldots,\widetilde{\eta}_{e,N}$;\\
$(e)$ follows from the fact that $\eta_{e,i},\widetilde{\eta}_{e,i}$ ($i \in\{1,\ldots,N\}$) are i.i.d. random variables and are independent of each other;\\
$(f)$ follows from conditions reduce entropy;\\
$(g)$ follows from for $i\in\{1,\ldots,N\}$, $\eta_{e,i}\sim \mathcal{N}(0, \sigma_{e}^2)$, $\widetilde{\eta}_{e,i}\sim \mathcal{N}(0, \widetilde{\sigma}_{e}^2)$, and the fact that $S$ is independent of $\eta_{e,1},\ldots,\eta_{e,N},\widetilde{\eta}_{e,1},\ldots,\widetilde{\eta}_{e,N}$, hence we have\\
{\small \begin{equation}
	\begin{split}
		h&(\eta_{e,1}+\lambda S)+\sum_{i=2}^{N}h(-\lambda \cdot\kappa^{\frac{i-1}{2}}\cdot S+\eta_{e,i})
		\leq\frac{1}{2}\sum_{i=1}^{N}\ln(2 \pi e(P\cdot\kappa^{i-1}+\sigma_{e}^{2})),
		\end{split}
\end{equation}
\begin{equation}
	\begin{split}
		h&(\widetilde{\eta}_{e,1}+\lambda S)+\sum_{i=2}^{N}h(-\lambda \cdot\kappa^{\frac{i-1}{2}}\cdot S+\widetilde{\eta}_{e,i})
		\leq\frac{1}{2}\sum_{i=1}^{N}\ln(2 \pi e(P\cdot\kappa^{i-1}+\widetilde{\sigma}_{e}^{2}));
		\end{split}
\end{equation}
}%
$(h)$ follows from the inequality $\ln(1+x)\leq x$;\\
$(i)$ follows from the sum of a geometric sequence.\\
Finally, substituting (\ref{1to1-h(s|Z)}) into (\ref{1to1-Ln-fl}), the upper bound of the information leakage rate is given by
{\small \begin{equation}\label{1to1-Ln-f-fl}
	\begin{aligned}
		L_{N}&=\frac{I(S;Z^{N},\widetilde{Z}^{N})}{N}=\frac{h(S)-h(S|Z^{N},\widetilde{Z}^{N})}{N}\\
		&\leq \frac{1}{2N}\left(\frac{(P+\sigma_{\eta}^{2})(\sigma_{e}^{2}+\widetilde{\sigma}_{e}^{2})}{\sigma_{e}^{2}\widetilde{\sigma}_{e}^{2}}\right)\left(1-\kappa^{N}\right),
	\end{aligned}
\end{equation}}%
which completes the proof.

\end{document}